\documentclass[a4paper,reqno]{amsart}

\usepackage{amsfonts,amssymb,amsmath,amsthm,gensymb,color,epic,eepic,float,fullpage}
\usepackage[mathscr]{euscript}
\usepackage[T1]{fontenc}
\usepackage{mathtools}
\let\savedegree\degree
\let\degree\relax
\usepackage{mathabx}
\let\degree\savedegree
\usepackage{comment}
\usepackage{stmaryrd}
\usepackage{hyperref}
\usepackage{dsfont}
\usepackage{bbm}
\usepackage{enumitem}
\usepackage{cite}
\usepackage{combelow}

\usepackage{tikz}
\usetikzlibrary{graphs,patterns,decorations.markings,arrows,matrix}
\usetikzlibrary{calc,decorations.pathmorphing,decorations.markings,
decorations.pathreplacing,patterns,shapes,arrows}
\tikzstyle{block} = [rectangle,draw,text width=10em,text centered,rounded corners,minimum height=4em]
\tikzstyle{line} = [draw, -latex']
\tikzset{arrow/.style={postaction={decorate,thick,decoration={markings,mark = at position #1 with {\arrow{>}}}}},arrow/.default=1}
\usepackage{youngtab}
\allowdisplaybreaks
\newtheorem{defn}{Definition}
\newtheorem{lem}{Lemma}
\newtheorem{thm}{Theorem}
\newtheorem{conj}{Conjecture}
\newtheorem{cor}{Corollary}
\newtheorem{rmk}{Remark}
\newtheorem{prop}{Proposition}
\newtheorem{ex}[thm]{Example}

\def\bbF{\mathbb{F}}
\def\bbQ{\mathbb{Q}}
\def\bbZ{\mathbb{Z}}

\def\leq{\leqslant}
\def\geq{\geqslant}

\def\a{\alpha}
\def\b{\beta}

\def\ve{\varepsilon}
\def\g{\gamma}
\def\D{\Delta}
\def\d{\delta}
\def\k{\kappa}
\def\l{\lambda}

\def\m{\mu}

\def\n{\nu}

\def\s{\sigma}
\def\r{\rho}

\def\z{\zeta}

\def\MD{\mathcal{D}}
\def\MC{\mathcal{C}}
\def\MF{\mathcal{F}}
\def\MG{\mathcal{G}}
\def\MI{\mathcal{I}}
\def\ML{\mathcal{L}}
\def\MM{\mathcal{M}}
\def\MO{\mathcal{O}}
\def\MR{\mathcal{R}}

\def\MH{\mathcal{H}}
\def\MM{\mathcal{M}}

\def\MRb{\widebar{\mathcal{R}}}

\def\MP{\mathcal{P}}
\def\ME{\mathcal{E}}

\def\MS{\mathcal{S}}

\def\MV{\mathcal{V}}
\def\MK{\mathcal{K}}
\def\MP{\mathcal{P}}
\def\MH{\mathcal{H}}

\def\RR{\mathrm{R}}

\def\ve{\varepsilon}

\def\deg{\text{deg}}

\def\dd{\text{d}}

\newcommand{\braket}[2]{\left\langle #1|#2\right\rangle}
\newcommand{\bra}[1]{\left\langle #1\right|}
\newcommand{\ket}[1]{\left|#1\right\rangle}

\newcommand{\bbrakket}[2]{\left\langle\langle #1|#2\right\rangle\rangle}
\newcommand{\bbra}[1]{\left\langle\langle #1\right|}
\newcommand{\kket}[1]{\left|#1\right\rangle\rangle}

\newcommand{\qg}{U_{q,t}(\overset{..}{\mathfrak{gl}}_1)}
\newcommand{\qgo}{U^\circ_{q,t}(\overset{..}{\mathfrak{gl}}_1)}
\newcommand{\qgop}{U^{\circ,+}_{q,t}(\overset{..}{\mathfrak{gl}}_1)}
\newcommand{\qgom}{U^{\circ,-}_{q,t}(\overset{..}{\mathfrak{gl}}_1)}
\newcommand{\qgopm}{U^{\circ,\pm}_{q,t}(\overset{..}{\mathfrak{gl}}_1)}
\newcommand{\qgp}{U^+_{q,t}(\overset{..}{\mathfrak{gl}}_1)}
\newcommand{\qgm}{U^-_{q,t}(\overset{..}{\mathfrak{gl}}_1)}
\newcommand{\qgg}{U^\geq_{q,t}(\overset{..}{\mathfrak{gl}}_1)}
\newcommand{\qgl}{U^\leq_{q,t}(\overset{..}{\mathfrak{gl}}_1)}

\newcommand{\ba}{\[\begin{aligned}~}
\newcommand{\ea}{\end{aligned}\]}

\newcommand{\Hilb}{\text{Hilb}}

\begin{document}

\title[]{Computing the $R$-matrix of the quantum toroidal algebra}

\address{Alexandr~Garbali, School of Mathematics and Statistics, University of Melbourne, Victoria 3010, Australia.}
\author{Alexandr Garbali and Andrei Negu\cb{t}}
\email{alexandr.garbali@unimelb.edu.au}

\address{Andrei~Negu\cb t, MIT, Department of Mathematics, Cambridge, MA, USA}
\address{Simion Stoilow Institute of Mathematics, Bucharest, Romania}
\email{andrei.negut@gmail.com}

   \begin{abstract}
    We consider the problem of the $R$-matrix of the quantum toroidal algebra $\qg$ in the Fock representation. Using the connection between the $R$-matrix $R(u)$ ($u$ being the spectral parameter) and the theory of Macdonald operators we obtain explicit formulas for $R(u)$ in the operator and matrix forms. These formulas are expressed in terms of the eigenvalues of a certain Macdonald operator which completely describe the functional dependence of $R(u)$ on the spectral parameter $u$. We then consider the geometric $R$-matrix (obtained from the theory of $K$-theoretic stable bases on moduli spaces of framed sheaves), which is expected to coincide with $R(u)$ and thus gives another approach to the study of the poles of the $R$-matrix as a function of $u$.
    
    \end{abstract}

\maketitle

\section{Introduction}

The quantum toroidal algebra $\qg$ is an important object in geometric representation theory, mathematical physics and algebraic combinatorics. In fact, arguably the most important appearance of this algebra lies at the intersection of these three fields, namely the action on Fock space:
\begin{equation}
\label{eqn:action intro}
\qg \curvearrowright \MF
\end{equation}
Specifically, $\MF$ can be thought of either as the $K$-theory group of the Hilbert scheme of points on $\mathbb{A}^2$, the Hilbert space of quantum mechanics for an arbitrary number of identical bosons, or the ring of symmetric functions in countably many variables; in each of these incarnations, the action \eqref{eqn:action intro} encapsulates the symmetries of the Fock space.

As $\qg$ is a Hopf algebra, one can use its (topological) coproduct to define actions:
\begin{equation}
\label{eqn:action intro 2} 
\qg \curvearrowright \MF \otimes \MF
\end{equation}
(and more generally on tensor products of arbitrarily many Fock spaces). But as is often the case in quantum algebra, the order of the tensor factors matters, and simply swapping the factors in \eqref{eqn:action intro 2} does \underline{not} commute with the $\qg$ action. To fix this, one constructs the universal $R$-matrix:
\begin{equation} 
\label{eqn:universal intro}
\MR \in \qg \ \widehat{\otimes} \ \qg
\end{equation}
and uses it to obtain a $\qg$-module intertwiner:
\begin{equation} 
\label{eqn:R-matrix intro}
R(u) : \MF \otimes \MF \rightarrow \MF \otimes \MF (u)
\end{equation}
The operators \eqref{eqn:R-matrix intro} satisfy the quantum Yang--Baxter equation:
\begin{align}\label{eq:YB intro}
R_{1,2}(u_2/u_1)
R_{1,3}(u_3/u_1)
R_{2,3}(u_3/u_2)
=
R_{2,3}(u_3/u_2)
R_{1,3}(u_3/u_1)
R_{1,2}(u_2/u_1)
\end{align}
where $R_{i,j}(u_j/u_i)$ acts in the $i$-th and $j$-th spaces of $\MF\otimes\MF\otimes\MF$. 
Let $N(u)$ be the vacuum-vacuum matrix element of $R(u)$ (which can be found in formula \eqref{eq:N}) and set:
\begin{equation} 
\label{eqn:RR intro}
\RR(u):=N(u)^{-1} R(u)
\end{equation}
Since the universal $R$-matrix \eqref{eqn:universal intro} is an infinite sum (closely related to the canonical tensor of the Hopf pairing between two halves of the quantum toroidal algebra $\qg$), one can obtain formulas for $R(u)$ as a \underline{power series} in $u$. While the coefficients of this power series do not admit closed formulas, there are many ways to compute them, both systematically and algorithmically. However, the main result of the present paper is that $\RR(u)$ is actually the expansion of a \underline{rational function} with prescribed poles. 

\medskip

\begin{thm}
\label{thm:main intro}

(cf. Corollary \ref{cor:poles}) The operator \eqref{eqn:RR intro} is the power series expansion of an operator whose coefficients are rational functions in $u$, with simple poles at $\{u = q^it^{-j}\}_{i,j \ge 1}$. \\

Moreover, we give formulas for $R(u)$ and $\RR(u)$ in terms of bosons and in the standard Heisenberg basis of $\MF$, in Propositions \ref{prop:R-bar} and \ref{prop:matrix coefficients}, respectively. 

\end{thm}

\medskip

Our techniques for proving the Theorem above are two-fold: on one hand, we heavily use the Macdonald polynomial basis of $\MF$, and express vertex operators in terms of this basis. On the other hand, we recall from \cite{FHHSY, Shuf} the use of the shuffle algebra in dealing with such vertex operators. We prove certain shuffle algebra formulas for the $R$-matrix $R(u)$, which will allow us to identify its functional dependence on the parameter $u$ and prove Theorem \ref{thm:main intro}. Although we believe certain key formulas that we need are new, the approach summarized in the present paragraph has been used in many works, such as \cite{FJMM_BA, Awata, Fukud, GdG, Shuf, R-matrix}, and in the closely related problem of the instanton $R$-matrix \cite{MO , S , Proch19 , Litv20}.

In the second half of the present paper, we present an alternative approach to proving the first paragraph of Theorem \ref{thm:main intro}, via the $K$-theoretic stable basis construction (which originated in \cite{MO} and was developed in \cite{AO, O1, O2, OS} and other works) \footnote{We will take an expositional point of view rather than a purely rigorous one, so we do not claim an alternative proof due to certain details that we leave out for the sake of brevity. Some of these details are straightforward and some are less so, and we will inform the reader which is which}. The starting point is the observation that $\MF$ is isomorphic to the (algebraic, equivariant) $K$-theory groups of the Hilbert scheme of points on $\mathbb{A}^2$, a beautiful idea in geometric representation theory that arose from several important directions (notably the actions of the Heisenberg algebra on the cohomology of Hilbert schemes, due to Grojnowski and Nakajima, and the Bridgeland-King-Reid-Haiman equivalence of categories between modules over wreath products and sheaves on the Hilbert scheme). Taking this idea one step further, one identifies $\MF \otimes \MF$ with the $K$-theory groups of moduli spaces of rank 2 sheaves on the plane. The latter are symplectic varieties, and so any Lagrangian correspondence involving the said moduli spaces will give rise to homomorphisms:
\begin{equation}
\label{eqn:stable intro}
\MF \otimes \MF \rightarrow \MF \otimes \MF \left( \frac {u_2}{u_1} \right)
\end{equation}
where $u_1,u_2$ are simply equivariant parameters. The ``correct'' Lagrangian correspondences are the stable bases constructed in \cite{AO, MO, O1, O2, OS}, which were defined so that the operators \eqref{eqn:stable intro} satisfy the quantum Yang-Baxter equation. Though the details have not been written down (to the authors' knowledge), it is overwhelmingly expected that the operators \eqref{eqn:R-matrix intro} and \eqref{eqn:stable intro} match under the identification $u = \frac {u_2}{u_1}$. If this is so, then the fact that $R(u)$ is a rational function in $u$ would follow almost immediately, as operators given by Lagrangian correspondences are rational functions in $\frac {u_2}{u_1}$ for general reasons.

The structure of the present paper is the following:

\begin{itemize}
    
    \item In Section \ref{sec:Fock}, we review Macdonald polynomials, the Fock space and vertex operators
    
    \item In Section \ref{sec:algebra}, we review the quantum toroidal algebra $\qg$ and its shuffle algebra incarnation, and prove Theorem \ref{thm:main intro}
    
    \item In Section \ref{sec:stable basis}, we recall stable bases for Hilbert schemes, and present the alternative approach to the $R$-matrix using the operators \eqref{eqn:stable intro}

\end{itemize}

We would like to thank Jean-Emile Bourgine, Boris Feigin, Jan de Gier, Yakov Kononov, Andrey Smirnov and Alexander Tsymbaliuk for many interesting discussions on the quantum toroidal algebra and $R$-matrices. A. G. gratefully acknowledges financial support from the Australian Research Council.
A.N. gratefully acknowledges NSF grants DMS-$1760264$ and DMS-$1845034$, as well as support from the Alfred P.\ Sloan Foundation and the MIT Research Support Committee.


\section{The Fock space and symmetric functions}
\label{sec:Fock}
In this Section we provide some basic information about symmetric functions, Fock space, Heisenberg algebra and Macdonald operators.

\subsection{Partitions}
Let $\MP$ be the set of all partitions $\m=(\m_1,\m_2,\dots)$,  $\m_j\geq \m_{j+1}$. The {\it length of a partition} $\m$, denoted $\ell(\m)$, is equal to the number of non-zero parts in $\m$. 
The {\it weight of a partition} $|\m|$ is given by the sum of all parts $|\m|=\m_1+\dots +\m_{\ell(\m)}$. The notation $\m\vdash j$ means that $\m$ is a partition of $j$, $|\m|=j$. Let $r$ be a positive integer and $\mu$ a partition. Define the {\it part multiplicity function} $m_r(\m)$ which counts the number of parts in $\m$ equal to $r$. The {\it part multiplicity vector} $m(\m)$ is defined as $m(\m)=(m_1(\m),m_2(\m),\dots,m_{\m_1}(\m),0,\dots)$. It is clear that we have a bijection between the set of all partitions $\MP$ and non-negative integer vectors of part multiplicities. Using this correspondence we can add partitions by component-wise addition of the corresponding part multiplicity vectors. We can also subtract partitions when the subtraction of the part multiplicities does not produce vectors with negative entries.
So for three partitions $\m,\n,\l$ the notation:
\ba
\m\pm \n =\l
\quad \text{means}\quad 
m(\m)\pm  m(\n) =m(\l)
\ea
This defines a partial order: $\n \subseteq \m$, if $m(\m)- m(\n)\in\mathbb{Z}_{\geq 0}^{\infty}$. The intersection of two partitions $\m\cap \n$ is the partition $\l=\m\cap \n$ such that $\l\subseteq \m$, $\l\subseteq \n$ and the differences $\m-\l$ and $\n-\l$ have no common parts. We will also use the dominance order on partitions $\l\geq \m$, which means that $\l$ and $\m$ have the same weight and $\sum_{i=1}^k\l_i \geq \sum_{i=1}^k \m_i$ for all $k>0$.

Let us introduce the short hand notations involving factorials of part multiplicities of a partition $\m$:
\ba
&\m!:=\prod_{a = 1}^{\infty} m_a(\m)! 
\ea 
We define an analogue of the binomial coefficient for arguments in $\MP$. Let $\m,\n \in \MP$, we write:
\begin{align}
\label{binom0}
{\m \brack \n} :=
  \begin{cases}
\frac{\m!}{(\m-\n)! \n!},
&\quad \text{if } \n\subseteq \m \\
0,  & \quad \text{otherwise}
 \end{cases}
\end{align}
The usual properties of binomial coefficients apply to \eqref{binom0}, in particular the symmetry:
\begin{align}
\label{binom}
{\m \brack \n} =
{\m \brack \m - \n}
\end{align}
We will encounter summations over partitions. These summations always run over the infinite set of all partitions $\m\in\MP$ which we will simply write as $\sum_\m$. The summands will typically involve binomial coefficients \eqref{binom0} and due to their vanishing condition the summands will have finite support.

A partition $\m$ is identified with its Young diagram (we follow French notation), namely a collection of lattice squares in the first quadrant which is stable under gravity toward the third quadrant. The {\it transpose} of $\m$ is the partition $\m'$ whose Young diagram is the reflection of the Young diagram of $\m$ across the $x = y$ diagonal. The boxes of a Young diagram, denoted by $\square$, are identified by their coordinates $\square=(i,j)$, where the row index $i \geq 1$ increases upwards and the column index $j \geq 1$ increases rightwards. The arm and leg functions $a_\m(\square)$, $l_\m(\square)$ are defined by:
$$
a_\m(\square) = \m_i-j,
\quad
l_\m(\square) = \m'_j-i
$$
Sums or products over $\square \in \m$ mean that $\square$ runs over all the boxes in the Young diagram of the partition $\mu$. This should not be confused with the notation $a\in \mu$, which will simply be shorthand for the natural number $a$ running over the multiset $\{\mu_1,\dots,\mu_{\ell(\mu)}\}$.

\subsection{Symmetric functions}
Let $\Lambda$ be the ring of symmetric functions in the alphabet $x=(x_1,x_2,\dots)$ over $\mathbb{Z}$ and 
$\Lambda_{\mathbb{F}}=\Lambda \otimes \mathbb{F}$ where $\mathbb{F}=\mathbb{Q}(q,t)$. The basic symmetric functions are the monomial symmetric function $m_\l(x)$ and the power sum symmetric function $p_\l(x)$, $\l\in\MP$, where:
\ba
p_\l(x) = p_{\l_1}(x) \dots p_{\l_{\ell(\l)}}(x),
\qquad
p_r(x) = \sum_{i} x_i^r
\ea
For an integer $r$ and a partition $\lambda$ we set:
\begin{align}
\label{eq:z}
z_\l(q,t):=\l! \prod_{r\in\l}r \frac{1-q^r}{1-t^r}
\end{align}
The Macdonald scalar product on $\Lambda_\mathbb{F}$ is defined by:
\begin{align}
\label{pscalar}
\braket{p_\l}{p_\m}_{q,t} = \d_{\l,\m} z_\lambda(q,t)
\end{align}
Let us now recall the Macdonald polynomials $P_{\l}(x;q,t)$, which are characterised by the following two conditions \cite[Ch. \textrm{VI}]{Macdonald}:
\ba
&P_{\l}(x;q,t) = m_\l(x) +\sum_{\m<\l} u_{\l,\m}(q,t)m_\m(x)\\
&\braket{P_\l}{P_\m}_{q,t} = 0, \qquad \text{for}~\l\neq\m
\ea
As shown in \textit{loc. cit.}, these conditions imply that:
\begin{equation}
\label{eqn:b lambda}
\braket{P_\l}{P_\l}_{q,t} = \frac 1{b_\l(q,t)}, \quad \text{where} \quad b_\l(q,t) = \prod_{{\Large{\square}}\,\in\l}
\frac{1-q^{a_{\square}(\l)} t^{l_\square(\l)+1}}{1-q^{a_\square(\l)+1} t^{l_\square(\l)}}
\end{equation}
Thus, the dual Macdonald polynomials $Q_{\l}(x;q,t)$ are:
\begin{align}
\label{QMac}
Q_{\l}(x;q,t)=b_{\l}(q,t)P_{\l}(x;q,t)
\end{align}
and the scalar product becomes:
\begin{align}
\label{PQscalar}
\braket{P_\l}{Q_\m}_{q,t} = \d_{\l,\m}
\end{align}
We will effectively work with two bases of symmetric functions: the power sum basis and the Macdonald basis. Let us define the transition coefficients between these bases:
\begin{align}
\label{gcoef}
P_{\l}(x;q,t)= \sum_{\m} g_{\l,\m}(q,t)p_{\m}(x)
\end{align}
Since $P_{\l}(x;q,t)$ and $p_{\m}(x)$ are both homogeneous functions, this means that $g_{\l,\m}(q,t)$ are non-zero when $|\l|=|\m|$, i.e.:
\begin{align}
\label{gcoef0}
g_{\l,\m}(q,t) = 0,
\qquad
|\l|\neq|\m|
\end{align}

\subsection{The Fock space}
Let $\MH_a$ be the Heisenberg algebra over $\mathbb{F}$, generated by $\{a_{n},a_{-n}\}$ for $n\in\mathbb{Z}_{>0}$ with the defining relations:
\begin{align}
\label{Ha1}
[a_r, a_s] = [a_{-r}, a_{-s}] = 0 \qquad \text{and} \qquad [a_r,a_{-s}]=\d_{r,s}r\frac{1-q^r}{1-t^r}
\end{align}
for all $r,s > 0$. By multiplying these generators we get:
\ba
a_\m =\prod_{r \in \m} a_{r},\qquad
a_{-\m} =\prod_{r \in \m} a_{-r},\qquad \m\in\MP
\ea
The representation of this Heisenberg algebra is given on the Fock space $\MF$ and the dual Fock space $\MF^*$ by the formulas:
\begin{equation}
\label{Hact}
\begin{split}
a_{-\l} \ket{a_\m} = \ket{a_{\m+\l}},
\qquad
a_\l \ket{a_\m} =z_\l(q,t) {\m \brack \l }  \ket{a_{\m-\l}}
\\
\bra{a_\m} a_\l  =\bra{a_{\m+\l}}  ,
\qquad
\bra{a_\m} a_{-\l}  = \bra{a_{\m-\l}} z_\l(q,t) {\m \brack \l }
\end{split}
\end{equation}
The scalar product reads:
\begin{align}
\label{ascalar}
\braket{a_\l}{a_\m} = \d_{\l,\m} z_\l(q,t)
\end{align}
We denote by $T_{\a,\b}^{\m,\n}$ the matrix elements of the generic basis element $a_{-\mu} a_{\nu}$ of the Heisenberg algebra:
\begin{align}
\label{amatrix}
T_{\a,\b}^{\m,\n}:=\bra{a_\a}a_{-\m} a_{\n} \ket{a_\b}
\end{align}
We set $T$ to be the matrix of coefficients $T_{\a,\b}^{\m,\n}$  for all partitions $\a,\b,\m,\n$. 
The matrix elements of $T$ are computed using \eqref{Hact} and the matrix elements of $T^{-1}$ are computed using standard binomial identities:
\begin{align}
\label{TT}
T_{\a,\b}^{\m,\n}=\d_{\a-\m ,\b- \n}\,
z_\a(q,t) 
z_\n(q,t)
{\b \brack \n} ,
\qquad
(T^{-1})^{\a,\b}_{\m,\n} 
&= \d_{\m-\a ,\n- \b}\frac{(-1)^{\ell(\m-\a)}}{z_\a(q,t) z_\n(q,t)} {\m \brack \a} 
\end{align}

As graded vector spaces, the ring of symmetric functions and the Fock space $\MF$ are isomorphic:
\ba
\iota: \MF\rightarrow \Lambda_\mathbb{F},
\qquad
\ket{a_\l}\mapsto \ket{p_\l}
\ea
Using the isomorphism $\iota$ we can view Macdonald functions as vectors in the Fock space. They could be defined using the transition coefficients $g_{\l,\m}(q,t)$ between Macdonald functions and the power sums \eqref{gcoef}:
\begin{align}
\label{eq:g}
\bra{P_{\l}} = \sum_{\m} g_{\l,\m}(q,t)\bra{a_{\m}},
\qquad
\ket{Q_{\l}} =b_\l(q,t) \sum_{\m} g_{\l,\m}(q,t)\ket{a_{\m}}
\end{align}
These vectors give us an orthonormal basis. We also define the normalized vectors:
\begin{align}
\label{braket_norm}
\bbra{a_\l} := \frac{1}{z_\l(q,t)}\bra{a_\l},
\qquad
\kket{a_\l} := \ket{a_\l}
\end{align}
so that we have:
\begin{align}
\label{ascalar2}
\bbrakket{a_\l}{a_\m} = \d_{\l,\m}  
\end{align}

\subsection{Macdonald operators}
Following \cite{Sh,FHHSY} we build the vertex operators:
\begin{align}
\label{eq:eta}
&\eta(z) := \exp\left(\sum_{r=1}^{\infty}  \frac{1-t^{-r}}{r}a_{-r} z^r\right)\exp\left(- \sum_{r=1}^{\infty} \frac{1-t^{r}}{r} a_{r} z^{-r}\right)
\\
\label{eq:xi}
&\xi(z) := \exp\left(- \sum_{r=1}^{\infty}\frac{1-t^{-r}}{r}(t / q)^{r/2} a_{- r}z^{ r}\right)
\exp\left(\sum_{r=1}^{\infty} \frac{1-t^r}{r} (t/q)^{r/2} a_{ r}z^{- r}\right)
\end{align}
The Fourier modes of these operators, denoted by $\eta_n$ and $\xi_n$:
\begin{align*}
    \eta(z) = \sum_{n\in \mathbb{Z}} \eta_n z^{-n},
    \quad
        \xi(z) = \sum_{n\in \mathbb{Z}} \xi_n z^{-n}
\end{align*}
have a well defined action on $\mathcal{F}$. A product of Heisenberg operators is {\it normally ordered} if all the positive modes are on the right and the negative modes are on the left. We use the standard notation $:a_\nu a_{-\mu}:\, =
:a_{-\mu} a_\nu:\, = a_{-\mu} a_\nu$. As such, the products $\eta(z) \eta(w)$ and $\xi(z)\xi(w)$ can be ordered according to:
\begin{align}
\label{eq:etacom}
&\eta(z) \eta(w)     
=
\zeta \left( \frac zw \right)
:\eta(z) \eta(w):\\
\label{eq:xicom}
&\xi(z)\xi(w)
=
\zeta \left( \frac wz \right)
:\xi(z)\xi(w):
\end{align}
where:
\begin{equation}
\label{eqn:def zeta}
\zeta(x) = \frac {(x-1)(x-q t^{-1})}{(x-q)(x-t^{-1})}
\end{equation}
\footnote{Note that the rational function \eqref{eqn:def zeta} is actually $\zeta (x^{-1})^{-1}$ in the notation of \cite{R-matrix}, but either choice gives rise to the same algebra structures.} In formulas \eqref{eq:etacom} and \eqref{eq:xicom}, the rational function $\zeta$ is expanded in non-negative powers of $z/w$. Using the vertex operators one can write the free field realizations of Macdonald operators. The basic Macdonald operators are the operators $\widehat{E}_1$ and $\widehat{E}^*_1$ which act on the Macdonald functions as:\footnote{These operators are denoted by $E_{q,t}$ and $E_{q^{-1},t^{-1}}$ in \cite[Ch. \textrm{VI}]{Macdonald} respectively.}
\begin{align}
    \widehat{E}_1 \ket{P_\lambda}= \sum_{i\geq 1} (q^{\lambda_i}-1) t^{-i} \ket{P_\lambda},
    \qquad
    \widehat{E}_1^* \ket{P_\lambda}= \sum_{i\geq 1} (q^{-\lambda_i}-1) t^{i} \ket{P_\lambda}
\end{align}
These operators are realized on the Fock space by the zero modes $\eta_0$ and $\xi_0$:
\begin{align}
    \eta_0 = (t-1) \widehat{E}_1 + 1,
    \quad
        \xi_0 = (t^{-1}-1) \widehat{E}^*_1 + 1
\end{align}
For any partition $\l \in \MP$, consider the ring homomorphism $\ve_\l : \Lambda_{\mathbb{F}} \rightarrow \mathbb{F}$ generated by:
\begin{align}
\label{eq:spec}
\ve_\l(x_i) = q^{\l_i} t^{-i},
\qquad  i \geq 1
\end{align}
The eigenvalue equations for $\eta_0$ and $\xi_0$ are:
\begin{align*}
\frac{t^{-1}}{1-t^{-1}}\eta_0 \ket{P_\lambda}= 
\ve_\lambda(e_1(x))\ket{P_\lambda},
\qquad
\frac{t}{1-t}\xi_0 \ket{P_\lambda}= 
\ve_\lambda(e_1(x^{-1}))\ket{P_\lambda}
\end{align*}
where the alphabet $(x^r)$ means $(x_1^r,x_2^r,\dots)$ for any integer $r$, and $e_1(x)$ is the first elementary symmetric function $e_1(x)= x_1 +x_2 +\dots$. More generally \cite{Sh} we define the operators:
\begin{align}
\label{En}
    \widehat{E}_n&:=\frac{1}{n!}\frac{t^{-n(n+1)/2}}{(1-t^{-1})^n} 
\oint
\prod_{i=1}^n \frac{\dd z_i}{2 \pi i z_i } 
\prod_{1 \leq i\neq j \leq n}
\frac{(z_i-z_j)}{(z_i-t^{-1} z_j)}
:\eta(z_1)\dots \eta(z_n):\\
\label{En_dual}
    \widehat{E}_n^*&:=\frac{1}{n!}\frac{t^{n(n+1)/2}}{(1-t)^n} 
\oint
\prod_{i=1}^n \frac{\dd z_i}{2 \pi i z_i } 
\prod_{1 \leq i\neq j \leq n}
\frac{(z_i-z_j)}{(z_i-t z_j)}
:\xi(z_1)\dots \xi(z_n):
\end{align}
whose action in the Macdonald basis is:
\begin{align}
    \widehat{E}_n \ket{P_\lambda}= 
\ve_\lambda(e_n(x))\ket{P_\lambda},
\qquad
    \widehat{E}_n^* \ket{P_\lambda}=
\ve_\lambda(e_n(x^{-1}))\ket{P_\lambda} \label{eqn:eigenvalues of e}
\end{align}
Since the operators $\widehat{E}_n$ are diagonal in the Macdonald basis, they commute with each other, and thus they generate a ring of operators which is isomorphic to the ring of symmetric functions (see \cite{FHHSY}).


\section{The quantum toroidal and shuffle algebras}
\label{sec:algebra}

We will now discuss the algebraic structure that governs the Fock space, and encompasses both the vertex operators $\eta(z)$, $\xi(z)$ and the Macdonald operators $\widehat{E}_{n}$ and $\widehat{E}^*_{n}$. We will focus on two isomorphic incarnations of this algebra:
\begin{equation}
\label{eqn:iso intro}
\qg \cong \MD \MS
\end{equation}
where the left-hand side is the quantum toroidal algebra that was featured in the Introduction, and the right-hand side is the (double extended) shuffle algebra $\MS$. We will find the latter incarnation to be more convenient for our computations. 

\subsection{The triangular decomposition}

The quantum toroidal algebra of type $\mathfrak{gl}_1$ (sometimes called the Ding-Iohara-Miki algebra \cite{DI, Miki}) can be presented by generators and relations, although we believe that it is more enlightening for our purposes to start from its triangular decomposition:
\begin{equation}
\label{eqn:triangular}
\qg = \qgp \otimes \qgo \otimes \qgm 
\end{equation}
The three tensor factors above can be referred to as the ``positive nilpotent'', ``Cartan'' and ``negative nilpotent'' subalgebras, by analogy with the theory of finite-dimensional Lie algebras. Moreover, the first two and the latter two factors fit together in analogues of the ``Borel subalgebras'':
\begin{align}
&\qgg = \qgp \bigotimes_{\bbF} \qgop  \label{eqn:quantum plus} \\
&\qgl = \qgm \bigotimes_{\bbF} \qgom \label{eqn:quantum minus}
\end{align}
where $\bbF = \bbQ(q,t)$. Our reason for stating the facts above before actually defining $\qg$ is that one can reconstruct the algebra structure from the bialgebra structures on the halves \eqref{eqn:quantum plus} and \eqref{eqn:quantum minus}, plus the bialgebra pairing:
\begin{equation}
\label{eqn:bialgebra pairing}
\qgg \otimes \qgl \xrightarrow{\langle \cdot, \cdot \rangle} \bbF 
\end{equation}
that we will recall in Subsection \ref{sub:bialgebra}. Thus, we find it more convenient to work backwards.

\subsection{The positive halves} Recall the rational function $\zeta(x)$ from \eqref{eqn:def zeta}.
\begin{defn}

Consider the algebras: 
\begin{align}
&\qgp = \bbF \langle e_n \rangle_{n \in \bbZ} \Big / \Big( \text{relation \eqref{eqn:rel quantum plus}} \Big) \label{eqn:def quantum plus} \\ 
&\qgm = \bbF \langle f_n \rangle_{n \in \bbZ} \Big / \Big( \text{relation \eqref{eqn:rel quantum minus}} \Big) \label{eqn:def quantum minus}
\end{align}
where $e(z) = \sum_{n \in \bbZ} \frac {e_n}{z^n}$, $f(z) = \sum_{n \in \bbZ} \frac {f_n}{z^n}$, and:
\begin{align}
&e(z) e(w) \zeta \left(\frac zw \right) = e(w) e(z) \zeta \left(\frac wz \right) \label{eqn:rel quantum plus} \\
&f(z) f(w) \zeta \left(\frac wz \right) = f(w) f(z) \zeta \left(\frac zw \right) \label{eqn:rel quantum minus}
\end{align}
To make sense of the relation above, one clears the denominators of the $\zeta$ functions and then identifies the coefficients of any $z^aw^b$ in the left and right-hand sides. Note that $\qgm = \qgp^{\emph{op}}$.

\end{defn}

One would like $\qgp$ and $\qgm$ to admit bialgebra structures, but as in the theory of finite-dimensional quantum groups, before doing so one needs to enlarge these algebras. To do so, define:
\begin{equation}
\label{eqn:cartan subalgebra}
\qgopm = \bbF\left[c,\psi_0^\pm,\psi_1^\pm,\psi_2^\pm,\dots \right] \Big/ c \text{ and } \psi_0^\pm \text{ invertible}
\end{equation}

Let us form the power series $\psi^\pm(z) = \sum_{r = 0}^{\infty} \frac {\psi_{\pm r}^{\pm}}{z^{\pm r}}$, and define the ``Borel'' subalgebras:
\begin{align}
&\qgg = \qgp \bigotimes_{\bbF}  \qgop \label{eqn:quantum ext plus} \\
&\qgl = \qgm \bigotimes_{\bbF} \qgom \label{eqn:quantum ext minus}
\end{align}
where the multiplication in the algebras above is governed by the relations:
\begin{align}
&
\psi^{+}(z)e(w) \z\left(\frac zw\right) = e(w)\psi^{+}(z)
\z\left(\frac wz\right)
   \label{eqn:rel quantum ext plus} \\
&\psi^{-}(z)f(w) \z\left(\frac wz\right) =f(w) \psi^{-}(z) \z\left(\frac zw\right) \label{eqn:rel quantum ext minus}
\end{align}
as well as the fact that $c$ is central. To make sense of the relations above, expand the rational functions in the right-hand sides as a power series in $w^{-1}$ and $w$, respectively.

\subsection{The non-negative halves} 
\label{sub:bialgebra}

The extended algebras $\qgg$ and $\qgl$ admit topological coproducts, completely determined by the following formulas: 
\begin{align}
&\Delta(c) = c \otimes c \label{eqn:cop quantum 0} \\
&\Delta(\psi^+(w)) = \psi^+(w) \otimes \psi^+(w/c^{(1)}) \label{eqn:cop quantum 1} \\
&\Delta(\psi^-(w)) = \psi^-(w/c^{(2)}) \otimes \psi^-(w) \label{eqn:cop quantum 2} \\
&\Delta(e(z))= e(z)\otimes 1 + \psi^+(z)\otimes e(z/c^{(1)}) \label{eqn:cop quantum 3} \\
&\Delta(f(z))= 1\otimes f(z) +  f(z/c^{(2)}) \otimes \psi^-(z) \label{eqn:cop quantum 4}
\end{align}
where $c^{(1)} = c \otimes 1$, $c^{(2)} = 1 \otimes c$. The formulas above induce bialgebra structures on $\qgg$ and $\qgl$. The counit $\varepsilon$ annihilates all generators $e_n, f_n$ and $\psi_{\pm r}^\pm$ for $r \neq 0$, and sends $c$ and $\psi^\pm_0$ to 1.

\begin{rmk}
\label{rem:hopf}

Note that $\qgg$ and $\qgl$ are actually Hopf algebras, and it is easy to see how to write down the antipode maps:
\begin{equation}
\label{eqn:antipode}
\qgg \xrightarrow{S} \qgg, \qquad \qgl \xrightarrow{S} \qgl
\end{equation}
from \eqref{eqn:cop quantum 1}--\eqref{eqn:cop quantum 3}. As we will not need the antipode in the present paper, we leave this as an exercise to the interested reader. Indeed, it is well known that the antipode is not an extra structure on a bialgebra, but a property which is determined by the product and coproduct.

\end{rmk}

\begin{defn}

A bilinear pairing:
\begin{equation}
\label{eqn:pairing}
\Big \langle \cdot, \cdot \Big \rangle : \qgg \otimes \qgl \longrightarrow \bbF
\end{equation}
is called a bialgebra pairing if it satisfies the properties:
\begin{align}
&\Big \langle a, bb' \Big \rangle = \Big \langle \Delta(a), b \otimes b' \Big \rangle \label{eqn:bialg 1} \\
&\Big \langle aa', b \Big \rangle = \Big \langle a \otimes a', \Delta^{\emph{op}}(b) \Big \rangle \label{eqn:bialg 2}
\end{align}
for all $a,a' \in \qgg$ and all $b,b' \in \qgl$. 

\end{defn}

Consider the assignments:
\begin{equation}
\label{eqn:pairing 1}
\Big \langle e(z), f(w) \Big \rangle = \frac {1}{\alpha} \cdot \delta \left(\frac zw \right) 
\end{equation}
where:
\begin{align}
\label{alpha}
\alpha : = \frac{1-qt^{-1}}{(1-q)(1-t^{-1})}
\end{align}
and:
\begin{equation}
\label{eqn:pairing 2}
\Big \langle \psi^+(z), \psi^-(w) \Big \rangle = \frac {\zeta \left(\frac wz \right)}{\zeta \left(\frac zw \right)} 
\end{equation}
(the right-hand side of \eqref{eqn:pairing 2} should be expanded as a power series in $|z| \gg |w|$) and $\langle c, - \rangle = \langle -, c \rangle = 1$. All other pairings between the $e$'s, $f$'s and $\psi$'s vanish. It is straightforward to see that the assignments above determine a bialgebra pairing, in accordance with \eqref{eqn:bialg 1}--\eqref{eqn:bialg 2}. Moreover, with respect to the $\mathbb{Z} 
\times \mathbb{Z}$ grading of the quantum toroidal algebra given by:
\begin{equation}
\label{eqn:grading on algebra}
\deg \ e_n = (1,n), \qquad \deg \ f_n = (-1,n), \qquad \deg \ \psi_n^\pm = (0,\pm n), \qquad \deg \ c = (0,0)
\end{equation}
we have $\langle a, b \rangle \neq 0$ only if $\deg \ a + \deg \ b = 0$.

\subsection{The full algebra} 

The main reason for introducing all the structure above is to give the following.

\begin{defn}

The Drinfeld double of $\qgg$ and $\qgl$ is defined as:
\begin{equation}
\label{eqn:double}
\qg = \qgg \otimes \qgl \Big / \left\{c \otimes 1 - 1 \otimes c , \psi_0^{+}\otimes 1 - 1 \otimes(\psi_0^{-})^{-1} \right\}
\end{equation}
where the mutiplication is controlled by the following relation:
\begin{equation}
\label{eqn:double rel}
a_1b_1 \Big \langle a_2,b_2 \Big \rangle = \Big \langle a_1,b_1 \Big \rangle b_2 a_2
\end{equation}
for all $a \in \qgg = \qgg \otimes 1 \subset \qg$ and $b \in \qgl = 1 \otimes \qgl \subset \qg$. In the formula above, we use Sweedler notation $\Delta(a) = a_1 \otimes a_2$ and $\Delta(b) = b_1 \otimes b_2$ for the coproduct, meaning that there are implied summation signs in front of the tensors and in front of the LHS and RHS of \eqref{eqn:double rel}.

\end{defn}

We leave it as an exercise to the interested reader that formula \eqref{eqn:double rel} implies the relations:
\begin{align}
&\psi^{-}(c z) e(w)\z\left(\frac zw\right)=e(w) \psi^{-}(c z)\z\left(\frac wz\right)  \label{eqn:rel double 1} \\
&\psi^{+}(c z) f(w)\z\left(\frac wz\right)=f(w) \psi^{+}(c z) \z\left(\frac zw\right)  \label{eqn:rel double 2} \\
&[e(z),f(w)]= \frac{1}{\alpha}\left( \d( \frac{cz}{w})\psi^-(w) - \d(\frac{cw}{z})\psi^+(z)\right) \label{eqn:rel double 3} 
\end{align}
The currents $\psi^{\pm}(z)$ can be viewed as exponential generating functions:
\begin{align}
\label{eqn:psi to h}
\psi^{\pm}(z)=\psi^{\pm}_0 \exp\left( \sum_{n>0}\frac{1}{n} (1-q^n)(1-t^{-n})(1-q^{-n}t^{n})h_{\pm n} z^{\mp n}\right)
\end{align}
Then the Drinfeld double relation implies that $\{h_{\pm n}\}_{n>0}$ generate a deformed Heisenberg algebra:
\begin{align}
[h_n,h_m] = [h_{-n},h_{-m}] = 0 \quad \text{and} \quad [h_n,h_{-m}]= \d_{m,n} \frac{n (c^n-c^{-n})}{(1-q^n)(1-t^{-n})(q^{-n}t^{n} - 1)} \label{eqn:rel double 4}
\end{align}
for all $m,n > 0$. Combining all the quadratic relations in the present Section, specifically \eqref{eqn:rel quantum plus}, \eqref{eqn:rel quantum minus}, \eqref{eqn:rel quantum ext plus}, \eqref{eqn:rel quantum ext minus}, \eqref{eqn:rel double 1}, \eqref{eqn:rel double 2}, \eqref{eqn:rel double 3}, \eqref{eqn:rel double 4} gives us the usual generators and relations of the quantum toroidal algebra. 

\begin{rmk}

In order for $\qg$ defined as in \eqref{eqn:double} to be an algebra, we need to make sense of products:
$$
a b a' b' a'' b'' \dots 
$$
for various $a,a',a'',\dots \in \qgg  \subset \qg$ and $b,b',b'',\dots \in \qgl \subset \qg$. In order for such a product to unambiguously describe an element of \eqref{eqn:double}, we need a way to convert products of the form $ab$ to linear combinations of products of the form $ba$, and vice versa. This is done by the following formula, which is equivalent to \eqref{eqn:double rel} by the properties of the antipode map \eqref{eqn:antipode}:
\begin{equation}
\label{eqn:double rel antipode}
ab = \Big \langle a_1,b_1 \Big \rangle b_2 a_2 \Big \langle a_3, S(b_3) \Big \rangle 
\end{equation}
where the right-hand side involves the Sweedler notation for the iterated coproduct: $\Delta^{(2)}(a) = a_1 \otimes a_2 \otimes a_3$ and $\Delta^{(2)}(b) = b_1 \otimes b_2 \otimes b_3$. The interested reader may check that there exists an antipode map (it is uniquely determined by the coproduct formulas) which satisfies all the Hopf algebra axioms. The reason why we choose to not write down the antipode map explicitly is that formula \eqref{eqn:double rel} leads to more elegant formulas than the equivalent formulas \eqref{eqn:double rel antipode} in the case of the quantum toroidal algebra. 

\end{rmk}

\subsection{The action I}
Let $u\in \mathbb{C}^*$ be the {\emph{spectral parameter}}. It is straightforward to check that there is an action \cite{FT,SV}: 
\begin{equation}
\label{eqn:action}
\qg \curvearrowright \Lambda_{\bbF}(u)
\end{equation}
generated by the assignments: 
\begin{align}
&c \mapsto (q/t)^{1/2} \\ 
&\psi_0^\pm \mapsto 1 \\
&e(z) \mapsto u  \xi(z)  \\
&f(z) \mapsto u^{-1}  \eta(z) \\
&
\psi^+(z)\mapsto \exp\left(-\sum_{r=1}^{\infty}\frac{1}{r} (1-t^r)(1-q^{-r} t^r) (q/t)^{r/2}  a_r z^{- r}\right)\\
&\psi^-(z)\mapsto \exp\left(\sum_{r=1}^{\infty}\frac{1}{r} (1-t^{-r})(1-q^{-r} t^{r})    a_{-r}z^{ r}\right)
\end{align}
Indeed, all one needs to show is that the formulas above satisfy relations \eqref{eqn:rel quantum plus}, \eqref{eqn:rel quantum minus}, \eqref{eqn:rel quantum ext plus}, \eqref{eqn:rel quantum ext minus}, \eqref{eqn:rel double 1}, \eqref{eqn:rel double 2}, \eqref{eqn:rel double 3}, \eqref{eqn:rel double 4}, and we leave this as an exercise to the interested reader. 

\subsection{The universal $R$-matrix} Given a Hopf algebra $A$, an element $\MR \in A \otimes A$ which satisfies:
\begin{align}
\label{deltaUR}
&\MR\D(g)=\Delta^{\text{op}}(g) \MR, \quad \forall g\in A
\\
\label{URR1}
&(\D\otimes 1)\MR=
\MR_{1,3}\MR_{2,3} \\
\label{URR2}
&(1 \otimes \D)\MR=
\MR_{1,3}\MR_{1,2}
\end{align}
is called a universal $R$-matrix. As a consequence of \eqref{deltaUR}, \eqref{URR1} and \eqref{URR2} we have the Yang--Baxter equation:
\begin{align}
    \label{eq:Univ_YB}
\MR_{1,2}\MR_{1,3}\MR_{2,3}
=
\MR_{2,3}\MR_{1,3}\MR_{1,2}
\end{align}
Drinfeld doubles such as \eqref{eqn:double} admit universal $R$-matrices:
\begin{equation}
\label{eqn:r-matrix abstract}
\mathcal{R} \in \qg \ \widehat{\otimes} \ \qg
\end{equation}
(the completed tensor product is necessary because the coproduct \eqref{eqn:cop quantum 0}--\eqref{eqn:cop quantum 4} is topological, i.e. its values are infinite sums) which are none other than the canonical tensors of the pairing \eqref{eqn:pairing}:
\begin{equation}
\label{eqn:r-matrix concrete}
\mathcal{R} = \sum_i a_i \otimes b^i
\end{equation}
where $\{a_i\}$ and $\{b^i\}$ are orthogonal bases of the subalgebras $\qgg$ and $\qgl$. If these subalgebras were finite-dimensional and the pairing were non-degenerate, then \eqref{eqn:r-matrix concrete} would be completely well-defined. However, both of the aforementioned properties fail for the quantum toroidal algebra $\qg$. The way to fix this failure is to note that the restriction of the pairing \eqref{eqn:pairing}:
\begin{equation}
\label{eqn:pairing restricted}
\Big \langle \cdot, \cdot \Big \rangle : \qgp \otimes \qgm \longrightarrow \bbF
\end{equation}
is indeed non-degenerate, and that the failure of non-degeneracy is simply due to the central elements $c$ and $\psi_0^\pm$ being in the kernel of the pairing. The way to fix this is to add central elements $d$ and $d^\perp$ to the quantum toroidal algebra, and extend the Hopf pairing in such a way that $\{d,d^\perp\}$ are orthogonal to $\{\log(c), \log(\psi_0^-)\}$. On this extended algebra, the pairing becomes non-degenerate and we have:
\begin{align}
\label{KR}
\MR=\MRb\MK
\end{align}
where we separate the ``Cartan part'' of the $R$-matrix:
\begin{align}
\label{MK}
&\MK=\exp\left( \sum_{r = 1}^{\infty} \frac {(1-q^r)(1-t^{-r})(1 - q^{-r}t^{r}) (h_{r} \otimes h_{-r})}r  \right)
e^{\log(c)\otimes d+d\otimes \log(c) +\log(\psi_0^-)\otimes d^{\perp}+d^{\perp}\otimes \log(\psi^-_0) }
\end{align}
from the canonical tensor of the restricted pairing \eqref{eqn:pairing restricted}, namely:
\begin{equation}
\label{eqn:r-matrix for quantum toroidal}
\MRb = \sum_i a_i \otimes a^i = 1 \otimes 1 + \alpha \sum_{n \in \bbZ} e_n \otimes f_{-n} + \text{higher terms}
\end{equation}
where $a_i$ and $a^i$ run over dual bases of $\qgp$ and $\qgm$. In the formula above, we have for all $i$:
\begin{equation}
\label{eqn:degrees}
\deg \ a_i + \deg \  a^i = 0
\end{equation}
where $\deg$ is the $\bbZ \times \bbZ$ grading on $\qg$ which is completely determined by $\deg \ e_n = (1,n)$ and $\deg \ f_n = (-1,n)$ and multiplicativity (see \eqref{eqn:grading on algebra}). Relation \eqref{eqn:degrees} is due to the fact that the pairing \eqref{eqn:pairing 1}, \eqref{eqn:pairing 2} only pairs non-trivially elements of opposite degrees. In the present paper, we will study the object \eqref{eqn:r-matrix for quantum toroidal}, and slightly abusively refer to it as ``the universal $R$-matrix of the quantum toroidal algebra''.

\subsection{The shuffle algebra} In the previous Subsections, we discussed the definition of the quantum toroidal algebra $\qg$, as well as its Hopf algebra structure and universal $R$-matrix. We saw that the quantum toroidal algebra acts on Fock space (via the vertex operators $\xi(z)$ and $\eta(z)$, but one thing which wasn't clear is how to incorporate the Macdonald operators $\widehat{E}_n$ and $\widehat{E}^*_n$ into this algebra action. Moreover, while \eqref{eqn:r-matrix for quantum toroidal} gives a recipe for defining the universal $R$-matrix, it is by no means an explicit formula. To remedy these issues, we will introduce an algebraic structure called the shuffle algebra, which will be isomorphic to the positive half of the quantum toroidal algebra:
$$
\MS \cong \qgp
$$
However, while $\qgp$ is defined by generators and relations, $\MS$ is completely explicit, and also leads to explicit formulas. To define the shuffle algebra, start from the vector space of symmetric polynomials in arbitrarily many variables:
\begin{equation}
\label{eqn:big shuf}
\MV = \bigoplus_{n = 0}^{\infty} \bbF(z_1,\dots,z_n)^{\text{symmetric}}
\end{equation}
Let us make $\MV$ into an algebra via the so-called shuffle product:
\begin{equation}
\label{eqn:shuf prod}
F(z_1,\dots,z_n) * F'(z_1,\dots,z_{n'}) = \text{Sym} \left[ \frac {F(z_1,\dots,z_n)F'(z_{n+1},\dots,z_{n+n'})}{n! n'!} \mathop{\prod_{1\leq a \leq n}}_{n < b \leq n+n'} \zeta \left( \frac {z_a}{z_b} \right) \right]
\end{equation}
where $\text{Sym}$ refers to the sum over all $(n+n')!$ permutations of the variables. It is a straightforward exercise to check that $\MV$ is an associative algebra with unit the function 1 in zero variables. 

\begin{defn}

The shuffle algebra is the subalgebra $\MS \subset \MV$ generated by:
\begin{equation}
\label{eqn:generators}
\{z_1^n\}_{n \in \bbZ}
\end{equation}

\end{defn}

It was shown in \cite{Shuf} that $\MS$ coincides with the subset of $\MV$ of rational functions of the form:
\begin{equation}
\label{eqn:wheel form}
F(z_1,\dots,z_n) = f(z_1,\dots,z_n) \prod_{1 \leq i \neq j \leq n}
\frac {1 - \frac {z_i}{z_j}}{\left(1 - \frac {qz_i}{z_j}\right) \left(1- \frac {z_i}{tz_j} \right)}
\end{equation}
where $f$ is a symmetric Laurent polynomial which satisfies the so-called wheel conditions of \cite{FHHSY}:
\begin{equation}
\label{eqn:wheel explicit}
f\left(\frac {xq}t, xq, x, z_4,\dots,z_n \right) = f\left(\frac {xq}t, \frac xt, x, z_4,\dots,z_n \right) = 0
\end{equation}
As such, it is very easy to note that the following explicit rational functions lie in $\MS$:
\begin{equation}
\label{eqn:comm}
F_n = \alpha^n \prod_{1 \leq i \neq j \leq n} \zeta \left(\frac {z_i}{z_j} \right) = \alpha^n \prod_{1 \leq i\neq j \leq n} \frac{(z_i-z_j)(z_i-q t^{-1}z_j)}{(z_i-q z_j)(z_i-t^{-1} z_j)}
\end{equation}
It was shown in \cite{FHHSY} that the rational functions $F_n$ commute with each other, and thus generate a commutative subalgebra of the shuffle algebra:
\begin{equation}
\label{eqn:comm}
\bbF[F_1,F_2,\dots] \subset \MS
\end{equation}
The connection between the shuffle algebra and the quantum toroidal algebra is the following result. 

\begin{prop}
\label{prop:iso}

The maps $f_n \mapsto z_1^n$ and $e_n \mapsto z_1^n$ extend to isomorphisms:
\begin{equation}
\label{eqn:iso}
\qgm \xrightarrow{\sim} \MS \qquad \text{and} \qquad \qgp \xrightarrow{\sim} \MS^{\emph{op}}
\end{equation}
respectively (above and henceforth, $\MS^{\emph{op}}$ refers to $\MS$ with the opposite algebra structure).

\end{prop}

\subsection{The coproduct on the shuffle algebra} The Hopf algebra structure on $\qg$ can be presented in terms of the shuffle algebra (see \cite{Shuf}, or \cite{R-matrix} for a more recent review in a related language to ours). Explicitly, first consider the extended shuffle algebras:
\begin{align}
&\MS^{\leq} = \MS \bigotimes_{\bbF} \qgom \Big / \text{relation \eqref{eqn:rel ext minus}} \label{eqn:ext minus} \\
&\MS^{\geq} = \MS^{\text{op}} \bigotimes_{\bbF} \qgop \Big / \text{relation \eqref{eqn:rel ext plus}} \label{eqn:ext plus} 
\end{align}
(compare with \eqref{eqn:cartan subalgebra}) where the relations are best written in terms of the formal series $\psi^\pm(w) = \sum_{r=0}^{\infty} \psi^{\pm}_{\pm r} w^{\mp r}$:
\begin{align}
&F(z_1,\dots,z_n) \psi^-(w) = \psi^-(w) F(z_1,\dots,z_n) \prod_{i=1}^n \frac {\zeta \left( \frac {z_i}w \right)}{\zeta \left( \frac w{z_i} \right)} \label{eqn:rel ext minus} \\
&G(z_1,\dots,z_n) \psi^+(w) = \psi^+(w) G(z_1,\dots,z_n) \prod_{i=1}^n \frac {\zeta \left( \frac w{z_i} \right)}{\zeta \left( \frac {z_i}w \right)} \label{eqn:rel ext plus} 
\end{align}
for all $F \in \MS$ and $G \in \MS^{\text{op}}$. The right-hand sides of the formulas above are expanded as power series in $w$ in the same direction as $\psi^\pm(w)$. By comparing \eqref{eqn:rel ext plus}--\eqref{eqn:rel ext minus} with \eqref{eqn:quantum ext plus}--\eqref{eqn:quantum ext minus}, we can see that the isomorphisms of Proposition \ref{prop:iso} extend to isomorphisms:
\begin{equation}
\label{eqn:iso geq}
\MS^{\leq} \xrightarrow{\sim} \qgl \qquad \text{and} \qquad \MS^{\geq} \xrightarrow{\sim} \qgg
\end{equation}
The reason for the extended algebras $\MS^\geq$ and $\MS^\leq$ is that they admit topological bialgebra structures via the following coproduct formulas:
\begin{align}
&\Delta(c) = c \otimes c \label{eqn:cop shuf 0} \\
&\Delta(\psi^+(w)) = \psi^+(w) \otimes \psi^+(w/c^{(1)}) \label{eqn:cop shuf 1} \\
&\Delta(\psi^-(w)) = \psi^-(w/c^{(2)}) \otimes \psi^-(w) \label{eqn:cop shuf 1.5} \\
&\Delta(F) = \sum_{k=0}^n \frac {F(z_1c^{(2)},\dots,z_k c^{(2)} \otimes z_{k+1},\dots,z_n) \cdot \prod_{i=1}^k \psi^-(z_i)}{\prod_{i=1}^k \prod_{j=k+1}^n \zeta \left(\frac {z_i}{z_j} \right)}\label{eqn:cop shuf 2} \\
&\Delta(G) = \sum_{k=0}^n \frac {\prod_{j=k+1}^n \psi^+(z_j) \cdot G(z_1,\dots,z_k \otimes z_{k+1}c^{(1)},\dots,z_nc^{(1)})}{\prod_{i=1}^k \prod_{j=k+1}^n \zeta \left(\frac {z_j}{z_i} \right)}  \label{eqn:cop shuf 3}
\end{align}
for any $F(z_1,\dots,z_n) \in \MS$ and any $G(z_1,\dots,z_n) \in \MS^{\text{op}}$. To make sense of the right-hand side of \eqref{eqn:cop shuf 2} and \eqref{eqn:cop shuf 3}, one expands the rational function as a power series in $|z_i| \ll |z_j|$ (for all $1\leq i \leq k$, $k < j \leq n$) and in every monomial of the resulting expression places all powers of $z_1,\dots,z_k$ to the left of the $\otimes$ sign, and all powers of $z_{k+1},\dots,z_n$ to the right of the $\otimes$ sign. It is easy to see that the coproduct \eqref{eqn:cop shuf 1}--\eqref{eqn:cop shuf 3} match the coproduct \eqref{eqn:cop quantum 0}--\eqref{eqn:cop quantum 4} under the isomorphisms \eqref{eqn:iso geq}.

\subsection{The action II}

With Proposition \ref{prop:iso} in mind, it should come to no surprise that one can present the action of the quantum toroidal algebra on Fock space in the language of the shuffle algebra (this was observed in \cite{FHHSY}). Specifically, composing the isomorphisms \eqref{eqn:iso} with the action \eqref{eqn:action} implies that:
\begin{align}
&F \in \MS \quad \text{acts on } \Lambda_{\bbF} \text{ via} \quad \frac {1}{n!}\oint_{\MC_n} \prod_{i=1}^n \frac{\dd z_i}{2 \pi i z_i } 
F(z_1,\dots,z_n) :\eta(z_1)\dots \eta(z_n): \label{eqn:action shuf plus} \\
&G \in \MS^{\text{op}} \quad \text{acts on } \Lambda_{\bbF} \text{ via} \quad \frac {1}{n!}\oint_{\MC_n} \prod_{i=1}^n \frac{\dd z_i}{2 \pi i z_i } 
G(z_1,\dots,z_n) :\xi(z_1)\dots \xi(z_n): \label{eqn:action shuf minus}
\end{align}
where $\MC_n$ is the contour given by $|z_1|=\dots=|z_n|=1$. The formulas above allow one to motivate the shuffle algebra as the abstract structure which governs the composition of the vertex operators $\eta(x)$ and $\xi(x)$, respectively. When $F = F_n$, the formulas above lead to a family of Macdonald operators similar to those in \eqref{En}-\eqref{En_dual}:
\begin{align}
\label{Fn}
\widehat{F}_{n}&:=\frac{1}{n!}
\oint_{\MC_n}^{|q| < 1 < |t|}
\prod_{i=1}^n \frac{\dd z_i}{2 \pi i z_i } 
F_n(z_1,\dots,z_n)
:\eta(z_1)\dots \eta(z_n):\\
\label{Fn_dual}
\widehat{F}_{n}^*&:=
\frac{1}{n!}
\oint_{\MC_n}^{|q| > 1 > |t|}
\prod_{i=1}^n \frac{\dd z_i}{2 \pi i z_i } 
F_n(z_1,\dots,z_n)
:\xi(z_1)\dots \xi(z_n):
\end{align}
where the superscript on the integral sign denotes the assumption on the sizes of
the parameters $q$ and $t$ that must be made in order to evaluate the contour integral. It is convenient to express the action of these operators in the Macdonald basis using their generating functions:
\begin{align}
\label{Fgen}
\widehat{F}(v):=\sum_{n=0}^{\infty}v^{-n} \widehat{F}_n,
\qquad
\widehat{F}^*(v):=\sum_{n=0}^{\infty}(-v)^{n} \widehat{F}_n^*
\qquad
\end{align}
In addition we define two functions:
\begin{align}
\label{eq:N}
N(v)&=\exp\left(\sum_{r>0} \frac{1}{r} \frac{1-q^rt^{-r}}{(1-q^r)(1-t^{-r})} v^{-r} \right)\\
\label{eq:Ns}
N^*(v)&=\exp\left(\sum_{r>0} \frac{1}{r} \frac{1-q^{-r}t^{r}}{(1-q^{-r})(1-t^{r})} v^{r} \right)
\end{align}
where we suppose $|q|,|t^{-1}|,|v^{-1}|<1$ in \eqref{eq:N} and $|q|,|t^{-1}|,|v^{-1}|>1$ in \eqref{eq:Ns}.
\begin{lem}
\label{lem:F}
The action of $\widehat{F}(v)$ and $\widehat{F}^*(v)$ in the Macdonald basis is given by:
\begin{align}
    \widehat{F}(v) \ket{P_\lambda}= 
    f_\l(v)\ket{P_\lambda},
\qquad
\widehat{F}^*(v) \ket{P_\lambda}=
f^*_\lambda(v)\ket{P_\lambda} \label{eqn:eigenvalue of f}
\end{align}
and the eigenvalues $f_\lambda (v),f^*_\lambda (v)$ are defined by:
\begin{align}
\label{eq:f-product}
f_\l(v)&:=N(v)\prod_{(i,j)\in \l}\frac{1-v^{-1} t^{-i+1}q^{j-1}}{1-v^{-1} t^{-i}q^{j}}\\
\label{eq:fs-product}
f_\l^*(v)&:=
N^*(v)\prod_{(i,j)\in \l}\frac{1-v  t^{i-1}q^{-j+1}}{1-v t^{i}q^{-j}}
\end{align}
\end{lem}
\begin{proof}
We will prove the formula for $\widehat{F}(v)$, as the one for $\widehat{F}^*(v)$ is analogous. The eigenvalue $f_\lambda(v)$ can be written in the exponential form:
\begin{align}
\label{eq:fexp}
	f_\lambda(v)
    &=\exp\left(\sum_{r>0} \frac{1}{r}\frac{1-q^rt^{-r}}{1-q^r}  v^{-r}t^r \ve_\lambda(p_r(x)) \right)
\end{align}
which can be shown with the help of the identity:
\begin{align}
\label{epsilon1}
\frac{1}{1-q^r}\ve_{\l}(p_r(x)) =
-q^{-r}\sum_{i=1}^{\ell(\l)}\sum_{j=1}^{\l_i} q^{j r} t^{-i r} + \frac{t^{-r}}{(1-t^{-r})(1-q^r)}
\end{align}
Let us define:
\begin{align}
&F(v) = \sum_{n=0}^{\infty}v^{-n} F_n \\
&E(v) = \sum_{n=0}^{\infty}v^{-n} \frac{t^{-n(n+1)/2}}{(1-t^{-1})^n} \prod_{1 \leq i\neq j \leq n} \frac{(z_i-z_j)}{(z_i-t^{-1} z_j)}
\end{align}
as elements of $\MS[[v^{-1}]]$. It was proved in \cite[Subsection 6.5]{Flags} that there exist elements $P_n \in \MS$ such that:
\begin{align}
&F(v) = \exp \left(\sum_{r=1}^{\infty} \frac {1-q^rt^{-r}}{(1-q^r)(1-t^{-r})} \frac {P_r}{r} v^{-r} \right) \label{eqn:exp f} \\
&E(v) = \exp \left(\sum_{r=1}^{\infty} \frac {(-1)^{r-1}}{t^r-1} \frac {P_r}{r} v^{-r} \right) \label{eqn:exp e}
\end{align}
Since formula \eqref{eqn:action shuf plus} yields an action, we conclude that there exist operators $\widehat{P}_n$ on $\Lambda_{\bbF}$ such that formulas \eqref{eqn:exp f} and \eqref{eqn:exp e} also hold with hats above $F,E$ and $P$, where:
$$
\widehat{E}(v) = \sum_{r=0}^{\infty} v^{-r} \widehat{E}_n
$$
However, formula \eqref{eqn:eigenvalues of e} then implies:
$$
\widehat{P}_n \ket{P_\lambda}= 
\ve_\lambda \left( (t^n-1)p_n(x) \right)\ket{P_\lambda}
$$
Combining this with the version of \eqref{eqn:exp f} with hats on top of $F$ and $P$ yields precisely \eqref{eqn:eigenvalue of f}.

\end{proof}

\subsection{Shuffle algebra formulas for the $R$-matrix} The shuffle algebra language yields a description of the universal $R$-matrix in the Fock representation. Specifically, it was shown in \cite[Theorem 4.16]{R-matrix} that the image of $\bar{\MR}$ under the action map:
$$
\qgp \ \widehat{\otimes} \ \qgm \longrightarrow \text{End}(\Lambda_{\bbF}(u)) \ \widehat{\otimes} \ \qgm \stackrel{\eqref{eqn:iso}} \cong \text{End}(\Lambda_{\bbF}(u)) \ \widehat{\otimes} \ \MS
$$
is given by the formula:
\begin{equation}
\label{eqn:half r-matrix}
\sum_{n = 0}^\infty \frac {(u\alpha)^{n}}{n!} \oint_{\MC_n} \prod_{i=1}^n \frac{\dd w_i}{2 \pi i w_i } :\xi(w_1)\dots \xi(w_n): \otimes \ S(w_1,\dots,w_n)
\end{equation}
where the formal series $S(w_1,\dots,w_n) \in \MS[[w_1^{\pm 1},\dots,w_n^{\pm 1}]]$ is defined by:
\begin{equation}
\label{eqn:formal series}
S(w_1,\dots,w_n) = \sum_{d_1,\dots,d_n \in \bbZ} w_1^{d_1} \dots w_n^{d_n} \text{Sym} \Big[ z_1^{d_1} \dots z_n^{d_n} \Big] \prod_{1\leq i \neq j \leq n} \zeta \left(\frac {z_i}{z_j} \right)
\end{equation}
We may then apply formula \eqref{eqn:action shuf plus} to conclude that the formal series $S(w_1,\dots,w_n)$ acts on $\Lambda_{\bbF}(u)$ as: 
\begin{equation}
\label{eqn:s action}
u^{-n} :\eta(w_1)\dots \eta(w_n): \prod_{1\leq i \neq j \leq n} \zeta \left(\frac {w_i}{w_j} \right)
\end{equation}
Combining this with \eqref{eqn:half r-matrix} we can compute $\bar{\MR}$ in a tensor product of Fock representations. Explicitly, let $R(u_2/u_1)$ be the image of the universal $R$-matrix $\MR$ under the action map:
\begin{equation}
    \label{eqn:r-matrix in fock}
\qg \ \widehat{\otimes} \ \qg \longrightarrow \text{End}(\Lambda_{\bbF}(u_1)) \ \widehat{\otimes} \ \text{End}(\Lambda_{\bbF}(u_2)), \qquad \MR \leadsto R(u_2/u_1)
\end{equation}
We denote by  $\bar{R}(u_2/u_1)$ and $K$ the corresponding images of $\bar{\MR}$ and $\MK$. From now on we will assume $|q| < 1 < |t|$ and $|u|>1$.
\begin{prop}
The matrix $R(u)$ is given by:
\begin{align}
\label{eq:RK}
R(u) = \bar{R}(u)K
\end{align}
where $\bar{R}(u) = \sum_{n = 0}^\infty u^{-n}\bar{R}_n$ is given by the formula:
\begin{equation}
    \label{eqn:r-matrix in Fock}
\bar{R}_n = \frac {\alpha^n}{n!} \oint_{\MC_n} \prod_{i=1}^n \frac{\dd w_i}{2 \pi i w_i }  :\xi(w_1)\dots \xi(w_n): \otimes :\eta(w_1)\dots \eta(w_n): \prod_{1\leq i \neq j \leq n} \zeta \left(\frac {w_i}{w_j} \right)
\end{equation}
and:
\begin{align}
\label{eq:KFock}
K=
\exp
\left(
\sum_{r>0}\frac{1}{r}\frac{(1-t^r)(1-q^r t^{-r})}{1-q^{r}}(t/q)^{r/2}  a_r\otimes a_{-r}
\right)
(q/t)^{\frac{1}{2}(d\otimes 1+1\otimes d)}
\end{align}
\end{prop}
We define a version of the path ordered exponential:
\begin{align}
    \label{eq:pexp}
    \MP \exp\left( \oint
\frac{\dd w}{2\pi i w} A(w)\right) 
=\sum_{n=0}^{\infty} \frac{1}{n!}
\oint_{\MC_n}
\prod_{i=1}^n \frac{\dd w_i}{2\pi i w_i} A(w_1)\dots A(w_n)
\end{align}
\begin{rmk}
$\bar{R}(u)$ is a path ordered exponential:
\begin{align}
\label{eq:R_exp}
\bar{R}(u)= 
\MP \exp\left(
\frac {\alpha}{u}
\oint
\frac{\dd w}{2\pi i w} \xi(w)\otimes \eta(w)\right)
\end{align}
\end{rmk}
\begin{proof}
Using \eqref{eq:etacom} and \eqref{eq:xicom} we can remove the rational function from the integrand in $\bar{R}_n$ in \eqref{eqn:r-matrix in Fock} together with the normal ordering:
\ba
\bar{R}_n = \frac {\alpha^{n}}{n!} 
\oint_{\MC_n}
\prod_{i=1}^n \frac{\dd w_i}{2\pi i w_i}
\xi(w_1)\dots \xi(w_n)\otimes \eta(w_1)\dots \eta(w_n)
\ea
Then we use the definition of $\MP \exp$ and arrive at \eqref{eq:R_exp}.
\end{proof}
\begin{rmk}
The partial vacuum-vacuum expectations and the vacuum-vacuum expectations of $\bar{R}(u)$ read:
\begin{align}
\label{RnF}
\left(\bra{\emptyset}\otimes - \right)
\bar{R}(u)
\left(\ket{\emptyset}\otimes -\right) &= 
\widehat{F}(u)
\\
\label{RN}
\left(\bra{\emptyset}\otimes \bra{\emptyset}\right)
\bar{R}(u)
\left(\ket{\emptyset}\otimes \ket{\emptyset}\right) &= 
N(u)
\end{align}
which is an immediate consequence of \eqref{eqn:r-matrix in Fock}. We should note that \eqref{RnF} was derived in \cite{FJMM_BA} by using commutation relations of the quantum toroidal algebra.
\end{rmk}
\begin{defn}
We define the normalized $R$-matrix:
\begin{align}
    \label{eq:RR}
    \RR(u):=\frac{1}{N(u)}R(u)
\end{align}
\end{defn}

\subsection{The $R$-matrix in terms of bosonic operators}
We focus on the $\bar{R}(u)$ and its expression in the Heisenberg basis \eqref{Hact}. We can write:
\begin{align}
\label{eq:Rprime}
&\bar{R}(u)
=
\sum_{\m,\r,\n,\s\in \MP}
[\bar{R}(u)]_{\m,\r}^{\n,\s}
\,a_{-\m}a_{\n}\otimes a_{-\r}a_{\s}\\
&
\label{eq:Rvanish}
[\bar{R}(u)]_{\m,\r}^{\n,\s}(u)=0,
\qquad |\m|+|\r|\neq |\n|+|\s|
\end{align}
(the vanishing condition \eqref{eq:Rvanish} follows from \eqref{eqn:degrees}). Introduce the operators:
\begin{align}
\label{c-ops}
c_{-r}:= 
1\otimes a_{-r}
-
(q/t)^{-r/2} a_{-r}\otimes 1, \qquad
c_r := 
1\otimes a_{r}
-
(q/t)^{-r/2} a_{r}\otimes 1
\end{align}
and $c_{\pm\m} = c_{\pm\m_1}\dots c_{\pm\m_{\ell(\m)}}$ for a partition $\m$. Recalling the explicit form of $\xi$ and $\eta$ we have:
\begin{align}
\label{etaxi}
:\prod_{i=1}^n \xi(z_i)\otimes \eta(z_i): ~=
\exp\left(\sum_{r>0}  \frac{1-t^{-r}}{r}
c_{-r} p_r(z) \right)
\exp\left( - \sum_{r>0} \frac{1-t^{r}}{r} 
c_{r} p_r(z^{-1}) \right)
\end{align}
where $p_r(z^{\pm 1}) = z_1^{\pm r} + \dots + z_n^{\pm r}$.

\begin{prop}\label{prop:R-bar}
$\bar{R}(u)$ expands in $c_{-\m}c_\n$ as:
\begin{align}
\label{Rcc}
\bar{R}(u) = 
 \sum_{\m,\n\in \MP} 
c_{-\m}c_\n
\left(  
 \sum_{\l\in \MP} b_\l(q,t)
 f_\l(u)
  \sum_{\a\subseteq \m\cap \n} 
    \frac{(-1)^{\ell(\a)}}{z_\a(q,t)}
     g_{\l,\m-\a}(q,t) g_{\l,\n-\a}(q,t)
\right)
\end{align}
\end{prop}
\begin{proof}
Recall formulas \eqref{Fn} with \eqref{eqn:r-matrix in Fock}:
\begin{align}
\label{Rintc}
\bar{R}_n&=\frac {1}{n!}
\oint_{\MC_n}
\prod_{i=1}^n \frac{\dd z_i}{2\pi i z_i}
F_n(z_1,\dots,z_n)
\exp\left(\sum_{r>0}  \frac{1-t^{-r}}{r}
c_{-r}p_r(z)\right)
\exp\left(- \sum_{r>0} \frac{1-t^{r}}{r} 
c_{r}p_r(z^{-1}) \right)
\\
\widehat{F}_n&=
\frac {1}{n!}
\oint_{\MC_n}
\prod_{i=1}^n \frac{\dd z_i}{2 \pi i z_i } 
F_n(z_1,\dots,z_n)
\exp\left(\sum_{r>0}  \frac{1-t^{-r}}{r}
a_{-r}p_r(z) \right)
\exp\left(- \sum_{r>0} \frac{1-t^{r}}{r} 
a_{r} p_r(z^{-1}) \right)
\end{align}
Clearly the expansion coefficients of $\bar{R}(u)$ in $c_{-\m}c_\n$ coincide with the expansion coefficients of $\widehat{F}(u)$ in $a_{-\m}a_\n$. 
If we denote these expansion coefficients by $X_{\m,\n}(u)$, then we have:
\begin{align}
\label{RFF}
\bar{R}(u) = \sum_{\m,\n\in \MP} 
X_{\m,\n}(u) c_{-\m}c_\n,
\qquad
\widehat{F}(u) = \sum_{\m,\n\in \MP} 
X_{\m,\n}(u) a_{-\m}a_\n
\end{align}
Now we calculate $X_{\m,\n}(u)$. Define the matrix elements $f_{\a,\b}(u)$ of $\widehat{F}(u)$ in the Heisenberg basis:
\begin{align}
\label{eq:fdef}
f_{\a,\b}(u):=\bra{a_\a}\widehat{F}(u)\ket{a_\b}
\end{align}
Recall the coefficients $T_{\a,\b}^{\m,\n}$ and $(T^{-1})_{\a,\b}^{\m,\n}$ from \eqref{TT}. The connection between $f_{\a,\b}(u)$ and $X_{\m,\n}(u)$ is:
\begin{align}
\label{eq:Ff}
f_{\a,\b}(u) =
\sum_{\m,\n\in \MP} 
T_{\a,\b}^{\m,\n}
X_{\m,\n}(u),\qquad
X_{\m,\n}(u) =
\sum_{\a,\b\in \MP} 
(T^{-1})_{\m,\n}^{\a,\b}
f_{\a,\b}(u)
\end{align}
The coefficients $f_{\a,\b}(u)$ are computed by writing $\widehat{F}(u)$ in its eigenbasis of Macdonald polynomials:
\begin{align}
\label{eq:Fmac}
\widehat{F}(u) = \sum_{\l\in \MP} f_\l(u) \ket{Q_\l}\bra{P_\l}
\end{align}
Let us insert \eqref{eq:Fmac} into \eqref{eq:fdef} and use the transition coefficients $g_{\a,\b}(q,t)$ from \eqref{gcoef} and the scalar product \eqref{ascalar}. We obtain:
\ba
f_{\a,\b}(u)
= \sum_{\l\in \MP} f_\l(u) 
\braket{a_\a}{Q_\l}\braket{P_\l}{a_\b}
&=
\sum_{\m,\n\in \MP}
 \sum_{\l\in \MP} b_\l(q,t)   f_\l(u)   
    g_{\l,\m}(q,t)   g_{\l,\n}(q,t)
\braket{a_\a}{a_\n}
 \braket{a_{\m}}{a_\b}\\
 &=
 \sum_{\l\in \MP}  b_\l(q,t) f_\l(u) 
   g_{\l,\a}(q,t)  g_{\l,\b}(q,t)
   z_{\a}(q,t)z_{\b}(q,t)
\ea
Substituting the formula above into the second equation of \eqref{eq:Ff}, and using \eqref{TT} to express $(T^{-1})^{\a,\b}_{\m,\n}$, yields:
\begin{align}
X_{\m,\n}(u) &=
\sum_{\a,\b\in \MP} 
 \d_{\m-\a ,\n- \b}{\m \brack \a} 
\frac{(-1)^{\ell(\m-\a)}}{z_\a(q,t) z_\n(q,t)}
 \sum_{\l\in \MP} b_\l(q,t) f_\l(u) 
   g_{\l,\a}(q,t)  g_{\l,\b}(q,t) z_{\a}(q,t)z_{\b}(q,t) \nonumber \\
    \label{eq:Xcoef}
 &= \sum_{\l\in \MP} b_\l(q,t) f_\l(u)
  \sum_{\a\subseteq \m\cap \n} 
    \frac{(-1)^{\ell(\a)}}{z_\a(q,t)}
     g_{\l,\m-\a}(q,t) g_{\l,\n-\a}(q,t)
\end{align}
\end{proof}
Since the operator $K$ in \eqref{eq:RK} does not depend on the spectral parameter $u$, Proposition \ref{prop:R-bar} shows that the dependence on $u$ of $R(u)$ is given by $f_\l(u)$ which is written explicitly in \eqref{eq:f-product}. For the normalized matrix $\RR(u)$ the spectral parameter dependence is given by the rational functions $N(u)^{-1} f_\l(u)$ and we conclude:

\begin{cor}
\label{cor:poles}

The operator $\RR(u)$ only has simple poles, and they are located at $\{u = q^jt^{-i}\}_{ i, j \geq 1}$. If we are interested in the restriction of $\RR(u)$ to the degree $\leq N$ part of $\Lambda_{\bbF}$, then the only poles we encounter are $u = q^jt^{-i}$ with $(i,j)$ among the boxes of partitions of weight $\leq N$.

\end{cor}

\subsection{Computing the matrix elements of $\RR(u)$}
We derive an explicit formula for the matrix elements of the normalized matrix $\RR(u)$. This formula uses Proposition \ref{prop:R-bar} and therefore the final answer depends on the transition coefficients $g_{\m,\n}(q,t)$ between the Macdonald functions and the power sums. We recall the normalized vectors of the Fock space \eqref{braket_norm} and define the matrix elements:
\begin{align}
\label{eq:Rme}
\RR_{\a,\b}^{\g,\d}(u) := \bbra{\a,\b} \RR(u) \kket{\g,\d}
\end{align}
where $\bbra{\a,\b} = \bbra{\a} \otimes \bbra{\b}$. Similarly to \eqref{eq:Rvanish} these matrix elements satisfy the vanishing condition:
\begin{align}
    \label{eq:rme_vanish}
    \RR_{\a,\b}^{\g,\d}(u)=0,
\qquad |\a|+|\b|\neq |\g|+|\d|
\end{align}
Due to \eqref{eq:Univ_YB} the matrix elements $\RR_{\a,\b}^{\g,\d}(u)$ satisfy the Yang--Baxter equation:
\begin{align}\label{eq:YB}
\sum_{a,b,c \in \MP} \RR_{\a,\a'}^{a,b}(u_2/u_1)
\RR_{a,\a''}^{\b,c}(u_3/u_1)
\RR_{b,c}^{\b',\b''}(u_3/u_2)
=
\sum_{a,b,c \in \MP}
\RR_{\a',\a''}^{b,a}(u_3/u_2)
\RR_{\a,a}^{c,\b''}(u_3/u_1)
\RR_{c,b}^{\b,\b'}(u_2/u_1)
\end{align}
for all fixed external indices $\a,\a',\a''$ and $\b,\b',\b''$ (the summations over partitions $a,b,c$ on both sides of \eqref{eq:YB} are finite due to \eqref{eq:rme_vanish}).

\begin{prop}
\label{prop:matrix coefficients}
Let $\a,\b,\g,\d$ be partitions.  The matrix elements of $\RR(u)$ are given explicitly by:
\begin{align}
\label{eq:Rprop}
\RR_{\a,\b}^{\g,\d}(u) = \sum_{\l\in \MP} b_\l(q,t) 
\prod_{(i,j)\in \l}\frac{1-u  t^{i-1}q^{-j+1}}{1-u t^{i}q^{-j}}
\sum_{\s \subseteq (\a+\b) \cap (\g+\d)} g_{\l,(\a+\b)-\s}(q,t)g_{\l,(\g+\d)-\s}(q,t) C_{\a,\b}^{\g,\d}(\s)
\end{align}
where:
\begin{align}
\label{eq:C}C_{\a,\b}^{\g,\d}(\s)  =  (-1)^{\ell(\a+\g+\s)}(q/t)^{|\b+\g|/2}
z_{(\g+\d)-\s}(q,t)
\sum_{\k \subseteq \b \cap \s} (q/t)^{-|\k|} 
{\g \brack \k}
{\d \brack \s-\k}
{(\a+\b)-\s \brack \b-\k}
\end{align}
\end{prop}
\begin{proof}
From \eqref{eq:RR} and \eqref{eq:RK} we have $\RR(u)=N(u)^{-1}\bar{R}(u) K$. Denote the matrix elements of $\bar{R}(u)$ and $K$ by:
\begin{align}
\label{eq:Rb_K}
\bar{R}_{\a,\b}^{\g,\d}(u) := \bbra{\a,\b} \bar{R}(u) \kket{\g,\d} ,\qquad
K_{\a,\b}^{\g,\d} := \bbra{\a,\b} K \kket{\g,\d}
\end{align}
We compute $\bar{R}(u)$ starting from \eqref{RFF}. Recall the definition of the operators $c_{\pm n}$ \eqref{c-ops} and then evaluate $\bbra{\a,\b}c_{-\m}c_\n\kket{\g,\d}$ using \eqref{braket_norm} and \eqref{amatrix}:
\ba
    \bbra{\a,\b}c_{-\m}c_\n\kket{\g,\d}
=&
\bbra{\a,\b}
\prod_{r\in \m} 
\left( 1\otimes a_{-r} -(q/t)^{-r/2} a_{-r}\otimes 1 \right)
\prod_{r\in \n} 
\left( 1\otimes a_{r} -(q/t)^{-r/2} a_{r}\otimes 1 \right)\kket{\g,\d}\\
=\sum_{\s\subseteq\m}
\sum_{\r\subseteq\n}
&(-1)^{\ell((\m+\n)-(\s+\r))}
(q/t)^{-\frac{1}{2}|\m+\n-\s-\r|}
{\m \brack \s}{\n \brack \r}
\bbra{\a,\b}a_{-(\m-\s)} a_{\n-\r}\otimes a_{-\s}a_\r\kket{\g,\d}\\
=\sum_{\s\subseteq\m}
\sum_{\r\subseteq\n}
&\frac{(-1)^{\ell((\m+\n)-(\s+\r))}
(q/t)^{-\frac{1}{2}|\m+\n-\s-\r|}}{z_\a(q,t)z_\b(q,t)}
{\m \brack \s}{\n \brack \r}
T^{\m-\s,\n-\r}_{\a,\g}
T^{\s,\r}_{\b,\d}
\ea
After writing the coefficients $T$ explicitly \eqref{TT} and using:
\begin{align}
\label{eq:zz}
z_{\n}(q,t) = z_{\n-\r}(q,t) z_{\r}(q,t) {\n \brack \r} 
\end{align}
we get:
\ba
\bar{R}_{\a,\b}^{\g,\d}(u)
 &= \sum_{\m,\n\in \MP} 
X_{\m,\n}(u) 
\sum_{\s\subseteq\m}
\sum_{\r\subseteq\n}
\frac{(-1)^{\ell((\m+\n)-(\s+\r))}
(q/t)^{-\frac{1}{2}|\m+\n-\s-\r|}}{z_\a(q,t)z_\b(q,t)}
{\m \brack \s}{\n \brack \r}
T^{\m-\s,\n-\r}_{\a,\g}
T^{\s,\r}_{\b,\d}\\
=
(-1)^{\ell(\a+\g)}
&\sum_{\m,\n\in \MP} 
X_{\m,\n}(u) 
\sum_{\s\subseteq\m}
\sum_{\r\subseteq\n}
\d_{\a-(\m-\s),\g-(\n-\r)}
\d_{\b-\s,\d-\r}
(q/t)^{-\frac{1}{2}|\m+\n-\s-\r|}
z_{\n}(q,t)
{\m \brack \s}{\g \brack \n-\r}{\d \brack \r}
\ea
The Kronecker delta $\d_{\b-\s,\d-\r}$ allows us to fix $\r=\d-(\b-\s)$ while the other Kronecker delta can be simplified:
\begin{multline}
\label{eq:Rme0}
\bar{R}_{\a,\b}^{\g,\d}(u)=(-1)^{\ell(\a+\g)}(q/t)^{\frac{1}{2}(|\a|-|\g|)}
\sum_{\substack{\m\subseteq (\a+\b)\\
\n\subseteq (\g+\d)}}
\d_{(\a+\b)-\m,(\g+\d)-\n} \, z_\n(q,t) X_{\m,\n}(u)  \\
\sum_{\s}  (q/t)^{|\s-\m|} 
{\m \brack \s} 
{\g \brack \n-(\d-(\b-\s))}
{\d \brack \b-\s}  
\end{multline}
Next we ``reverse'' the summations over $\m$ and $\n$ by replacing everywhere $\m \rightarrow (\a+\b)-\m$ and $\n \rightarrow (\g+\d)-\n$. After that the Kronecker delta becomes $\d_{\m,\n}$ and allows us to remove the summation over $\n$:
\begin{align}
\label{eq:Rme1}
\bar{R}_{\a,\b}^{\g,\d}(u)=
\sum_{\m}
z_{(\g+\d)-\m}(q,t) X_{(\a+\b)-\m,(\g+\d)-\m}(u)
\tilde{C}_{\a,\b}^{\g,\d}(\m)
\end{align}
where:
\begin{align}
\label{eq:Ctilde}
\tilde{C}_{\a,\b}^{\g,\d}(\m):=    (-1)^{\ell(\a+\g)}(q/t)^{-\frac{1}{2}(|\a|+|\g|)-|\b|}\sum_{\s}  (q/t)^{|\s+\m|}  
{(\a+\b)-\m \brack \s} 
{\g \brack \m-(\b-\s)}
{\d \brack \b-\s}  
\end{align}
The evaluation of $K$ produces:
\begin{align}
\label{eq:Kme}
K_{\a,\b}^{\g,\d}=\d_{\a+\b,\g+\d} (q/t)^{|\a+\d|/2} {\g \brack \a} \prod_{r\in \g-\a} (1-q^r t^{-r})
\end{align}
We need to compute $\RR_{\a,\b}^{\g,\d}(u) =N(u)^{-1} \sum_{\s,\r} \bar{R}_{\a,\b}^{\s,\r}(u) K_{\s,\r}^{\g,\d}$. After inserting the explicit form of the coefficient $X$ \eqref{eq:Xcoef} and simplifying we find:
\begin{multline}
\label{eq:RX}
\RR_{\a,\b}^{\g,\d}(u)=
\sum_{\l\in \MP} b_\l(q,t)  \prod_{(i,j)\in \l}\frac{1-u  t^{i-1}q^{-j+1}}{1-u t^{i}q^{-j}}
\sum_{\k\subseteq (\a+\b)\cap (\g+\d) } g_{\l,(\a+\b)-\k}(q,t)
g_{\l,(\g+\d)-\k}(q,t) \\
\sum_{\m\subseteq\k} \frac{(-1)^{\ell(\k-\m)}z_{(\g+\d)-\m}(q,t)}{z_{\k-\m}(q,t)}
\sum_{\s,\r}\tilde{C}_{\a,\b}^{\s,\r}(\m)K_{\s,\r}^{\g,\d}
\end{multline}
The first line of this expression matches with \eqref{eq:Rprop} up to the factor $C_{\a,\b}^{\g,\d}(\k)$. The remaining part of the proof consists of showing that:
\begin{align}
\label{eq:C-identity}
\sum_{\m\subseteq\k} \frac{(-1)^{\ell(\k-\m)}z_{(\g+\d)-\m}(q,t)}{z_{\k-\m}(q,t)}
\sum_{\s,\r}\tilde{C}_{\a,\b}^{\s,\r}(\m)K_{\s,\r}^{\g,\d}
=C_{\a,\b}^{\g,\d}(\k),
\qquad 
\forall \k\subseteq (\a+\b)\cap(\g+\d)
\end{align}
On the left hand side of \eqref{eq:C-identity} one needs to use \eqref{eq:zz} to rewrite $z_{(\g+\d)-\m}(q,t)=z_{(\g+\d)-\k+(\k-\m)}(q,t)$ in terms of $z_{(\g+\d)-\k}$, $z_{\k-\m}$ and a binomial coefficient. Dividing both sides of the resulting equation by  $z_{(\g+\d)-\k}$ we obtain an identity 
satisfied by two Laurent polynomials in $q/t$ (expand the product $\prod_{r\in \g-\s} (1-(q/t)^r)$ appearing in $K_{\s,\r}^{\g,\d}$ in powers of $q/t$). The rest is an exercise in computing summations on the left hand side of the resulting identity using standard manipulations with binomial coefficients.
\end{proof}

\section{Stable bases and Hilbert schemes}
\label{sec:stable basis}

We will now provide an alternative viewpoint on the $R$-matrix, through a geometric construction known as the stable basis (which originated in \cite{MO}, and was developed in \cite{AO, O1, O2, OS} and other works). We will start by reviewing the general theory, focusing on the crucial aspects, all the while glossing over technical details. Thus the following presentation should not be taken as a completely rigorous discussion (however, we will provide appropriate references), but rather as an invitation to a more in-depth study that does not require prior knowledge of algebraic geometry. Specifically, we provide:

\begin{itemize}

\item in Subsections \ref{sub:alg geom} - \ref{sub:localization} : a quick review of certain basic aspects of algebraic geometry, which would be contained in most introductory courses (such as the first three chapters of \cite{H})

\item in Subsections \ref{sub:hilbert} - \ref{sub:torus acts on M} : an introduction to Hilbert schemes and moduli spaces of framed sheaves on the plane, following \cite{N}

\item in Subsections \ref{sub:symplectic} - \ref{sub:normal} : a discussion of certain symplectic algebraic varieties equipped with torus actions, following \cite{MO}

\item in Subsections \ref{sub:stable} - \ref{sub:poles} : a definition (modulo technical details) of the $K$-theoretic stable basis (\cite{AO,O1,O2}) and its usefulness in constructiong trigonometric $R$-matrices (\cite{OS})

\end{itemize}

\subsection{Algebraic varieties} 
\label{sub:alg geom}

An affine algebraic variety is, by definition, the subset of $\mathbb{C}^m$ cut out by a certain collection of polynomial equations:
\begin{equation}
\label{f's}
V = \Big\{(x_1,\dots,x_m) \in \mathbb{C}^m, \ f_1(x_1,\dots,x_m) = \dots = f_k(x_1,\dots,x_m) = 0 \Big\}
\end{equation}
A (Zariski) open subset of $V$ refers to the following subset, defined for any polynomials $g_1,\dots,g_l$:
\begin{equation}
\label{g's}
U = \Big\{(x_1,\dots,x_m) \in V, \ g_1(x_1,\dots,x_m) \neq 0 \text{ or } \dots \text{ or } g_l(x_1,\dots,x_m) \neq 0 \Big\}
\end{equation}
Our main objects of interest are algebraic varieties $X$ over $\mathbb{C}$, i.e. spaces covered by open subsets:
\begin{equation}
\label{patch}
U \subset X
\end{equation}
which are all isomorphic to open subsets of affine algebraic varieties, as in \eqref{g's}. Moreover, for any open subsets $U,U' \subset X$, the gluing maps between $U \cap U'$ viewed as a subset of $U$ and $U \cap U'$ viewed as a subset of $U'$ are given by rational functions. A subvariety $Z \subset X$ is given by imposing further polynomial equations on the open subsets \eqref{patch}, in a way which is compatible with the gluing maps between any two open subsets. The ``smallest'' subvarieties of $X$ are points:
$$
p \in X
$$
which are sometimes called closed points \footnote{This is because, in the language of scheme theory, the word ``point'' also refers to irreducible subvarieties of $X$ of positive dimension}. Many familiar geometric notions, such as dimension (always over $\mathbb{C}$) and tangent spaces, apply to algebraic varieties, see \cite{H} for details.

The algebraic treatment of varieties starts by considering their coordinate rings, i.e. the rings of algebraic functions on these varieties. For instance, if $V \subset \mathbb{C}^m$ is the affine algebraic variety cut out by the polynomial equations \eqref{f's}, then the Nullstellensatz tells us that its coordinate ring is:
$$
R_V = \mathbb{C}[x_1,\dots,x_m] \Big / \sqrt{(f_1,\dots,f_k)}
$$
Going one step further, if $U \subset V$ is the open subset defined by $g\neq 0$ for some polynomial $g$, then:
$$
R_U = (R_V)_{g} = \left\{ \frac a{g^n}, \ a\in R_V, n \in \mathbb{N} \right\}
$$
As for general algebraic varieties, they have local coordinate rings $R_U$ for every open subset \eqref{patch}. The gluing maps between two open subsets $U$ and $U'$ lead to ring homomorphisms:
\begin{equation}
\label{glue}
R_U \rightarrow R_{U \cap U'} \leftarrow R_{U'}
\end{equation}
If $Z \subset X$ is a subvariety, then the restriction maps $R_U \rightarrow R_{U\cap Z}$ are quotient maps of rings. 

\subsection{Sheaves}
\label{sub:sheaves}

Since the study of algebraic geometry is mostly centered around coordinate rings of varieties, and these rings are commutative, it makes sense to study module theory over such rings.

\begin{defn}

A (coherent) sheaf $\MF$ on $X$ is an object which consists of $R_U$-modules $F_U$, as $U$ goes over all open subsets \eqref{patch}. When defining a sheaf, one needs to specify how to glue the various modules $F_U$.

\end{defn}

We refer to \cite[Chapter 2]{H} for the complete set of axioms that defines a coherent sheaf, and simply work with the intuitive ``definition'' above. The basic example of a coherent sheaf is the structure sheaf of $X$, i.e. the one defined by:
$$
F_U = R_U
$$
for all open subsets $U$, with gluing maps provided by \eqref{glue}. The structure sheaf is usually denoted by $\MO_X$. More generally, a coherent sheaf $\MF$ for which the $R_U$-modules $F_U$ are free for all small enough open subsets $U$, is called locally free. Locally free sheaves can be identified with the familiar notion of vector bundles from topology. However, not all sheaves are locally free, as the following examples show. 

\begin{ex}
\label{ideals}

Let $X = \mathbb{A}_{\mathbb{C}}^2$. Its ring of polynomial functions is $\mathbb{C}[x,y]$, so any module over this ring will correspond to a coherent sheaf over $X$. In particular, ideals:
\begin{equation}
\label{ideal on c2}
I \subset \mathbb{C}[x,y]
\end{equation}
are examples of such modules/sheaves. If we assume that the colength:
$$
n := \dim_{\mathbb{C}} \mathbb{C}[x,y]/I
$$
is finite, then $I$ is not a free module, and so the corresponding coherent sheaf is not locally free.

\end{ex}

\begin{ex}
\label{skyscraper}

Given a point $p \in X$, the skyscraper sheaf $\mathbb{C}_p$ is defined by the property that:
$$
(\mathbb{C}_p)_U = \begin{cases} \mathbb{C} &\text{if } p \in U \\ 0 &\text{if }p \notin U \end{cases}
$$
If $p \in U$, then the action of $f \in R_U$ on $(\mathbb{C}_p)_U$ is simply given by multiplication with the scalar $f(p) \in \mathbb{C}$.

\end{ex}

If $Z \subset X$ is any subvariety, then any coherent sheaf $\MF$ on $X$ can be restricted to $Z$, i.e. we can define the coherent sheaf:
$$
\MF|_Z \text { on } Z
$$
whose defining $R_{U \cap Z}$ modules are (for any open subset $U \subset X$):
$$
F_U \bigotimes_{R_U} R_{U \cap Z}
$$
endowed with the gluing maps inherited from those of $\MF$. The extremal case of restriction is when $Z = p$ is a point, in which case $\MF|_p$ is called the fiber of $\MF$ at $p$.

\subsection{Operations with sheaves} 
\label{sub:operations}

Given coherent sheaves $\MF$ and $\MG$, their direct sum and tensor product:
\begin{equation}
\label{sum and product}
\MF \oplus \MG \quad \text{and} \quad \MF \otimes \MG
\end{equation}
are defined by taking the direct sums and tensor products, respectively, of all the $R_U$-modules $F_U$ and $G_U$. A map $f : \MF \rightarrow \MG$ of coherent sheaves on $X$ consists of $R_U$-module maps $F_U \rightarrow G_U$ for all open subsets $U \subset X$, which are suitably compatible under gluing. The kernel and image of $f$ are also coherent sheaves \footnote{We note that while the $R_U$-modules $(\text{Ker }f)_U$ are simply $\text{Ker } (F_U \rightarrow G_U)$, the $R_U$-modules $(\text{Im }f)_U$ are actually defined by a procedure known as sheafification applied to $\text{Im }(F_U \rightarrow G_U)$, see \cite{H}.}, denoted by:
$$
\text{Ker }f \quad \text{and} \quad \text{Im }f
$$
If $\text{Ker }f = 0$ then $f$ is called injective, while if $\text{Im }f = \MG$ then $f$ is called surjective. 

\begin{defn}

A short exact sequence of coherent sheaves:
\begin{equation}
\label{ses}
0 \rightarrow \MF \xrightarrow{f} \MG \xrightarrow{g} \MH \rightarrow 0
\end{equation}
is one such that the morphism $f$ is injective, the morphism $g$ is surjective, and $\emph{Im }f = \emph{Ker }g$. 

\end{defn}

\begin{defn}
\label{K-theory}

The (0-th) algebraic $K$-theory group of $X$ is the abelian group:
$$
K(X)
$$
generated by symbols $[\MF]$ for all coherent sheaves $\MF$ on $X$, modulo the relations:
\begin{equation}
\label{k-theory rel}
[\MF] - [\MG] +[\MH] = 0
\end{equation}
for all short exact sequences \eqref{ses}.

\end{defn}

When $X$ is smooth, $K(X)$ is a ring with multiplication defined by derived tensor product of sheaves. The fact that we need the word ``derived'' in the previous sentence stems from the fact that the usual tensor product $\MF \otimes \MG$ does not preserve the relations \eqref{k-theory rel} unless one of $\MF$ or $\MG$ is locally free. However, if $X$ is smooth, any coherent sheaf $\MF$ admits a resolution by a finite complex of locally free sheaves $\dots \rightarrow \ME_i \rightarrow \ME_{i-1} \rightarrow \cdots$, so we may define the derived tensor product:
\begin{equation}
    \label{eqn:derived tensor product}
[\MF] \stackrel{L}{\otimes} [\MG] = \sum_i (-1)^i [\ME_i \otimes \MG]
\end{equation}
As an exercise in commutative algebra, one may check that the definition above does not depend on any of the choices made, and endows $K(X)$ with a commutative ring structure.

\begin{ex} As the previous paragraph shows, it is very elegant to work with locally free sheaves (i.e. vector bundles), but one often needs to work with general coherent sheaves. For example, given any subvariety $Z \subset X$, its structure sheaf $\MO_Z$ is a coherent sheaf on $X$ \footnote{This is because for any open subset $U \subset X$, the ring $R_{U \cap Z}$ is a $R_U$-module via the restriction morphism $R_U \rightarrow R_{U \cap Z}$.} and thus an element in $K$-theory:
$$
[\MO_Z] \in K(X)
$$
In the extremal case, when $Z = p$ is a point, this construction recovers skyscraper sheaves: $\MO_p = \mathbb{C}_p$. 

\end{ex}

\subsection{Torus actions} 
\label{torus actions}

Throughout the present paper, an algebraic torus is $T = (\mathbb{C}^*)^k$ for various natural numbers $k$. An algebraic variety $X$ is called a $T$-variety if it is endowed with an action:
$$
T \curvearrowright X
$$ 
by which we mean that general elements $(t_1,\dots,t_k) \in T = (\mathbb{C}^*)^k$ act on $T$-invariant open subsets $U \subset X$ (and implicitly on the local coordinate rings $R_U$) by rational functions in $t_1,\dots,t_k$. 

\begin{ex} 
\label{action on c2}

The standard action:
\begin{equation}
\label{action on plane}
\mathbb{C}^* \times \mathbb{C}^* \curvearrowright \mathbb{A}_{\mathbb{C}}^2, \qquad (t_1,t_2) \cdot (x,y) =  \left( \frac x{t_1}, \frac y{t_2} \right)
\end{equation}
corresponds to the following action on the ring of functions:
\begin{equation}
\label{action on functions}
\mathbb{C}^* \times \mathbb{C}^* \curvearrowright \mathbb{C}[x,y], \qquad (t_1,t_2) \cdot f(x,y) = f(t_1x,t_2y)
\end{equation}

\end{ex}

\begin{defn} Let $X$ be a $T$-variety. A coherent sheaf $\MF$ on $X$ is called $T$-equivariant if the torus $T$ acts on the $R_U$-modules $F_U$ in a way which is compatible with the action of $T$ on the rings $R_U$:
$$
t\cdot (r  f) = (t \cdot r)  (t \cdot f), \qquad \forall t \in T, r \in R_U, f \in F_U
$$
for all $T$-invariant open subsets $U \subset X$. 

\end{defn}

\begin{ex}
\label{fixed ideal}

Combining Examples \ref{ideals} with \ref{action on c2}, it is not hard to see that a finite colength ideal \eqref{ideal on c2} is $\mathbb{C}^* \times \mathbb{C}^*$ equivariant if and only if it is a monomial ideal:
\begin{equation}
\label{monomial ideal}
I_\lambda = (x^{\lambda_1}, x^{\lambda_2}y, \dots, x^{\lambda_d}y^{d-1},y^d)
\end{equation}
for some partition $\lambda_1 \geq \lambda_2 \geq \dots \geq \lambda_d > 0$. The number $\lambda_1+\dots+\lambda_d$ is the colength of the ideal $I_\lambda$.

\end{ex}

\subsection{Equivariant $K$-theory} 
\label{sub:equiv k-theory}

If $X$ is a $T$-variety, then all the notions of Subsection \ref{sub:operations} make sense for $T$-equivariant coherent sheaves $\MF, \MG, \MH$. In particular, morphisms of sheaves $\MF \rightarrow \MG$ are called $T$-equivariant if the corresponding $R_U$-module homomorphisms $F_U \rightarrow G_U$ respect the $T$-action. A short exact sequence is called $T$-equivariant if all the sheaves and morphisms in its definition are $T$-equivariant.

\begin{defn}

The $T$-equivariant (0-th) algebraic $K$-theory group of $X$, denoted by: 
$$
K^T(X)
$$
is defined as in Definition \ref{K-theory}, but allowing only $T$-equivariant coherent sheaves and short exact sequences.

\end{defn}
 
Besides the structures discussed in Subsection \ref{sub:operations} (including derived tensor product if $X$ is smooth), equivariant $K$-theory has the additional structure of a module over the representation ring of $T$:
\begin{equation}
\label{rep ring 1}
\text{Rep}_T \curvearrowright K^T(X)
\end{equation}
where $\text{Rep}_T$ denotes the set of linear combinations of finite-dimensional (algebraic) $T$-representations, made into a ring via direct sum and tensor product; alternatively, one may think of $\text{Rep}_T$ as the ring of characters of such $T$-representations. Explicitly, if $V$ is a finite-dimensional $T$-representation and $\MF$ is a $T$-equivariant coherent sheaf on $X$, then $V \otimes \MF$ has a natural structure of a $T$-equivariant coherent sheaf on $X$. If $T = (\mathbb{C}^*)^k$, then:
\begin{equation}
\label{rep ring 2}
\text{Rep}_T = \mathbb{C}[q^{\pm 1}_1,\dots,q^{\pm 1}_k]
\end{equation}
where $q_i$ denotes the one-dimensional $T$-representation corresponding to the character $(t_1,\dots,t_k) \mapsto t_i$. 

\subsection{The localization theorem}
\label{sub:localization}

Let $X$ be a smooth algebraic variety, which entails the existence of the tangent (respectively cotangent) locally free sheaf $\text{Tan } X$ (respectively $\text{Tan}^\vee X$). We further assume that $X$ has a $T$-action whose fixed point locus $X^T$ is finite. We will write:
$$
K^T(X)_{\text{loc}} = K^T(X) \bigotimes_{\text{Rep}_T} \text{Frac}(\text{Rep}_T)
$$
for the localized equivariant $K$-theory groups of $X$. This formally means that elements of $K^T(X)_{\text{loc}}$ can be multiplied not only with elements of the ring $\text{Rep}_T$ (i.e. Laurent polynomials in the elementary characters $q_1,\dots,q_k : T \rightarrow \mathbb{C}^*$) but with elements of the field:
$$
\mathbb{F} = \text{Frac}(\text{Rep}_T)
$$
(i.e. rational functions in $q_1,\dots,q_k$). Thus, $K^T(X)_{\text{loc}}$ is a $\mathbb{F}$-vector space. 

\begin{thm} \label{loc thm} (\cite{TT}) If $X$ is a smooth $T$-variety, we have an isomorphism of $\mathbb{F}$-vector spaces:
$$
K^T(X)_{\emph{loc}} \cong \bigoplus_{p \in X^T} \mathbb{F} \cdot [p]
$$
where $[p] = [\mathbb{C}_p]$ is the class of the skyscraper sheaf at the torus fixed point $p \in X^T$ (we assume that there are finitely many torus fixed points). The isomorphism above is explicitly given by:
\begin{equation}
\label{loc exp}
[\MF] \mapsto \sum_{p \in X^T} \frac {\MF|^L_p}{\wedge^\bullet (\emph{Tan}^\vee_p X)} \cdot [\mathbb{C}_p]
\end{equation}
for all $T$-equivariant coherent sheaves $\MF$, where the derived fiber (see \eqref{eqn:derived tensor product}) is defined as:
$$
\MF|^L_p = [\MF] \stackrel{L}\otimes [\mathbb{C}_p]
$$
Thus, $\MF|^L_p$ can be thought of as an alternating sum of $T$-equivariant vector spaces (supported at $p \in X$), as can the total exterior power $\wedge^\bullet (\emph{Tan}^\vee_p X)$. The fraction in \eqref{loc exp} is defined as the ratio of the alternating sums of the $T$-characters in said vector spaces, and this ratio is an element of $\mathbb{F}$.

\end{thm}

The main substance of Theorem \ref{loc thm} lies in the fact that the classes of the skyscraper sheaves $\mathbb{C}_p$ form a $\mathbb{F}$-basis of $K^T(X)_{\text{loc}}$. Once one accepts this fact, formula \eqref{loc exp} is a simple exercise which follows from:
\begin{equation}
\label{eqn:skyscraper restriction}
\mathbb{C}_p|^L_{p'} = \begin{cases} \wedge^\bullet (\text{Tan}^\vee_p X) &\text{if }p = p' \\ 0 &\text{otherwise} \end {cases}
\end{equation}
(the equality above is one of characters of alternating sums of $T$-equivariant vector spaces). 

\subsection{Hilbert schemes}
\label{sub:hilbert}

Let us now apply the notions above to a particularly interesting algebraic variety.

\begin{defn} The Hilbert scheme of $n$ points on the affine plane $\mathbb{A}_\mathbb{C}^2$ is:
\begin{equation}
\label{hilb 1}
\emph{Hilb}_n = \left \{ \text{colength } n \text{ ideals } I \subset \mathbb{C}[x,y] \right \}
\end{equation}

\end{defn}

To interpret $\Hilb_n$ as an algebraic variety, we consider the following alternative description. Given endomorphisms $X,Y$ of a vector space $V$, a vector $v \in V$ is called cyclic if $\text{span}\{P(X,Y)v\} = V$, as $P$ goes over all non-commutative polynomials with coefficients in $\mathbb{C}$. Then we have:
\begin{equation}
\label{hilb 2}
\Hilb_n = \left \{ (X, Y, v) \in \text{Mat}_{n \times n} \times \text{Mat}_{n \times n} \times \mathbb{C}^n \text{ such that } [X,Y] = 0 \text{ and }v \text{ is cyclic} \right\} /GL_n
\end{equation}
where $GL_n$ acts on $X,Y$ by conjugation and on $v$ by left multiplication. To get from an ideal $I$ as in \eqref{hilb 1} to a triple $(X,Y,v)$ as in \eqref{hilb 2}, we fix an identification $\mathbb{C}^n = \mathbb{C}[x,y]/I$, and let $X$ and $Y$ be the operators of multiplication by $x$ and $y$. Meanwhile, $v$ is simply $1 \text{ mod } I$. To get from a triple $(X,Y,v)$ as in \eqref{hilb 2} to an ideal $I$ as in \eqref{hilb 1} is also a straightforward exercise, one which we leave to the interested reader. The presentation \eqref{hilb 2} is manifestly an algebraic variety: the matrix entries of $X,Y,v$ constitute an $n^2+n^2+n$ dimensional affine space; from this we cut out an algebraic variety by imposing the commutation relations $[X,Y] = 0$, while requiring that $v$ is cyclic determines an open subset. \footnote{Care must be taken to properly formulate the operation of taking the $GL_n$ quotient in the context of algebraic varieties. The appropriate language here is that of geometric invariant theory, which deals with the existence and properties of such quotients.}

\subsection{The torus action on $\Hilb$} 
\label{sub:torus acts on Hilb} 

We will consider the torus $T = \mathbb{C^*} \times \mathbb{C}^* \times \mathbb{C}^*$ and the action: 
\begin{equation}
\label{action hilb}
T \curvearrowright \Hilb_n
\end{equation}
defined as follows: the first two factors of $\mathbb{C}^*$ act on ideals in a way which is inherited from their action on $\mathbb{C}[x,y]$ from \eqref{action on functions}. Meanwhile, the third factor of $\mathbb{C}^*$ acts trivially on $\Hilb_n$, so its presence is cosmetic at the moment, but will play a key role when we study moduli spaces of higher rank sheaves. In the language of \eqref{hilb 2}, the action \eqref{action hilb} is given by:
$$
(t_1,t_2, \xi)\cdot (X,Y,v) = (t_1X,t_2Y,\xi v)
$$
A point of $\Hilb_n$ is fixed by $T$ precisely when the corresponding ideal is $T$-equivariant, which as we saw in Example \ref{fixed ideal}, is precisely asking that the ideal be monomial. Therefore, we conclude that the fixed points of $\Hilb_n$ are in one-to-one correspondence with partitions of weight $n$:
$$
\Hilb_n^T = \{I_\lambda, \lambda \vdash n\}
$$
Let $q_1,q_2,u$ be the usual elementary characters of the three factors of the torus \eqref{action hilb}. If we write:
$$
\mathbb{F}_1 = \mathbb{C} (q_1,q_2,u)
$$
then the localization Theorem \ref{loc thm} gives us:
$$
K^T(\Hilb_n)_{\text{loc}} \cong \bigoplus_{|\lambda| = n} \mathbb{F}_1 \cdot [\lambda]
$$
where $[\lambda] = [\mathbb{C}_{I_\lambda}]$. To get the Fock space, we simply need to let $n$ run over all non-negative integers:
$$
\Hilb = \bigsqcup_{n=0}^\infty \Hilb_n, \qquad K^T(\Hilb) = \bigoplus_{n=0}^\infty K^T(\Hilb_n)
$$
and we conclude that:
$$
K^T(\Hilb)_{\text{loc}} \cong \bigoplus_{\lambda} \mathbb{F}_1 \cdot [\lambda]
$$
By work of Haiman, it is natural to identify:
\begin{equation}
\label{identify hilb}
K^T(\Hilb)_{\text{loc}} \cong \Lambda_{\mathbb{F}}(u)
\end{equation}
by sending $[\lambda]$ to the modified Macdonald polynomials $\widetilde{H}_\lambda$ (close relatives of $P_\lambda$, and we refer to \cite{Ha} for details). The notation $\Lambda_{\mathbb{F}}(u)$ simply refers to $\Lambda_{\mathbb{F}}$, but we insert the parameter $u$ in the notation to keep track of the third $\mathbb{C}^*$ action. This action is trivial on $\text{Hilb}_n$, but it will play a role in the next Subsection.

\subsection{Moduli of higher rank sheaves} 
\label{sub:moduli}

The Hilbert scheme is, perhaps after the Grassmannian, one of the most fundamental examples of a moduli space: an algebraic variety whose points parameterize algebro-geometric objects of a different nature (in the present case, ideals in $\mathbb{C}[x,y]$). As we have seen in Example \ref{ideals}, ideals can be interpreted as coherent sheaves on $\mathbb{A}_{\mathbb{C}}^2$. As such, they have rank 1, meaning that for any open subset $U \subset \mathbb{A}_{\mathbb{C}}^2$ which misses finitely many points, the ideal is a rank 1 free module over the local coordinate ring of $U$. One can therefore ask if there exists a moduli space of rank $r$ sheaves on $\mathbb{A}_{\mathbb{C}}^2$, and the answer is sadly, no. However, the following closely related object exists and behaves nicely.

\begin{defn} Let $\infty \subset \mathbb{P}_{\mathbb{C}}^2$ be the line at infinity, so $\mathbb{A}_{\mathbb{C}}^2 = \mathbb{P}_{\mathbb{C}}^2 \backslash \infty$. There exists an algebraic variety $\MM(r)$, whose points are in one-to-one correspondence with pairs:
\begin{equation}
\label{framed sheaf}
(\MF, \phi)
\end{equation}
where $\MF$ is a rank $r$ coherent sheaf on $\mathbb{P}_{\mathbb{C}}^2$, and $\phi$ is an isomorphism:
\begin{equation}
\label{framing}
\phi : \MF|_\infty \xrightarrow{\sim} \MO_\infty^{\oplus r}
\end{equation}
The isomorphism $\phi$ is called a framing of $\MF$, and a pair \eqref{framed sheaf} is called a framed sheaf. We will write $\MM(r)_n \subset \MM(r)$ for the connected component of framed sheaves whose second Chern class is $n$.

\end{defn}

The gist of the construction above is that a framed sheaf is trivial (i.e. free) in a neighborhood of $\infty \subset \mathbb{P}_{\mathbb{C}}^2$, but it might have some non-trivial structure away from $\infty$. In particular, for general algebraic reasons, a rank 1 framed sheaf embeds inside the structure sheaf:
\begin{equation}
    \label{eqn:ideal in p2}
\mathcal{I} \subset \mathcal{O}_{\mathbb{P}^2_{\mathbb{C}}}
\end{equation}
Because of the framing, the inclusion above is an equality near $\infty$, but on the open subset $\mathbb{A}_{\mathbb{C}}^2 = \mathbb{P}_{\mathbb{C}}^2 \backslash \infty$, it corresponds to an ideal:
\begin{equation}
    \label{eqn:ideal in a2}
I \subset \mathbb{C}[x,y]
\end{equation}
hence $\MM(1) \cong \Hilb$. With this example in mind, it should be of no surprise that $\MM(r)_n$ admits the following alternate presentation for any $r$, by analogy with \eqref{hilb 2}:
\begin{multline}
\label{adhm}
\MM(r)_n = \Big \{ (X, Y, A, B) \in \text{Mat}_{n \times n} \times \text{Mat}_{n \times n} \times \text{Mat}_{n \times r} \times \text{Mat}_{r \times n} \\  \text{ such that } [X,Y] = AB \text{ and }A \text{ is cyclic} \Big\} /GL_n
\end{multline}
where $GL_n$ acts by conjugation on $X$ and $Y$ and by left (resp. right) multiplication on $A$ (resp. $B$). We call the matrix $A$ cyclic if $\text{span}\{P(X,Y)A \cdot \mathbb{C}^r\} = V$, as $P$ goes over all non-commutative polynomials with coefficients in $\mathbb{C}$. We refer to \cite[Chapter 2]{N} for a detailed description on how to connect a framed sheaf with a quadruple $(X,Y,A,B)$ as above.

\subsection{The torus action on $\MM(r)$} 
\label{sub:torus acts on M}

There is an action of $T_r = \mathbb{C}^* \times \mathbb{C}^* \times (\mathbb{C}^*)^{r}$ on $\MM(r)$, where the first two factors act on the sheaf $\MF$ (through their action on the projective plane $\mathbb{P}_{\mathbb{C}}^2$) while the last $r$ factors act by multiplying the framing isomorphism \eqref{framing} with diagonal matrices. In terms of the presentation \eqref{adhm}, this action is given by:
$$
(t_1,t_2,\xi_1,\dots,\xi_r) \cdot (X,Y,A,B) = (t_1X,t_2Y,AD,t_1t_2D^{-1}B), \quad \text{where } D = \text{diag}(\xi_1,\dots,\xi_r)
$$
The fixed points of this action are all direct sums of monomial ideals, i.e. for various partitions $\lambda^1,\dots,\lambda^r$:
\begin{equation}
\label{monomial sheaf}
\MF_{\lambda^1,\dots,\lambda^r} = \MI_{\lambda^1} \oplus \dots \oplus \MI_{\lambda^r}
\end{equation}
where the inclusion $\MI_{\mu} \subset \MO_{\mathbb{P}_\mathbb{C}^2}$ is defined as the identity near $\infty$, and as the inclusion $I_{\mu} \subset \mathbb{C}[x,y]$ on the open subset $\mathbb{A}_\mathbb{C}^2 = \mathbb{P}_\mathbb{C}^2 \backslash \infty$ (see \eqref{eqn:ideal in p2}--\eqref{eqn:ideal in a2}). The fact that $\MI_{\mu} = \MO_{\mathbb{P}_\mathbb{C}^2}$ in a neighborhood of $\infty$ determines a canonical framing of $\MF_{\lambda^1,\dots,\lambda^r}$, which makes \eqref{monomial sheaf} into a framed sheaf. More specifically, we have:
$$
\MM(r)_n^{T_r} = \{ \MF_{\lambda^1,\dots,\lambda^r}, \text{ where }\lambda^1, \dots, \lambda^r \in \MP \text{ satisfy } |\lambda^1|+\dots+|\lambda^r| = n\}
$$
Let $q_1,q_2,u_1,\dots,u_r$ be the usual elementary characters of the factors of $T_r$. If: 
$$
\mathbb{F}_r = \mathbb{C} (q_1,q_2,u_1,\dots,u_r)
$$
then the localization Theorem \ref{loc thm} gives us:
$$
K^{T_r}(\MM(r)_n)_{\text{loc}} \cong \bigoplus_{|\lambda^1|+\dots+|\lambda^r| = n} \mathbb{F}_r \cdot [\lambda^1,\dots,\lambda^r]
$$
where $[\lambda^1,\dots,\lambda^r]$ is shorthand for $[\mathbb{C}_{\MF_{\lambda^1,\dots,\lambda^r}}]$. As we sum over all natural numbers $n$, we have:
$$
K^{T_r}(\MM(r))_{\text{loc}} \cong \bigoplus_{\lambda^1,\dots,\lambda^r} \mathbb{F}_r \cdot [\lambda^1,\dots,\lambda^r]
$$
Thus, we may identify:
\begin{equation}
\label{identify m}
K^{T_r}(\MM(r))_{\text{loc}} \cong \Lambda_{\mathbb{F}}(u_1) \otimes \dots \otimes \Lambda_{\mathbb{F}}(u_r)
\end{equation}
by sending $[\lambda^1,\dots,\lambda^r]$ to the tensor product of modified Macdonald polynomials $\widetilde{H}_{\lambda^1} \otimes \dots \otimes \widetilde{H}_{\lambda^r}$. 

\begin{ex} 
\label{ex decomposition}

Instead of asking for the fixed point set of $\MM(r)$ under the whole $T_r$ action, one may ask for its fixed point set under a subtorus. In particular, for any decomposition $r = r_1 + r_2$, we may consider the one-parameter subtorus:
\begin{equation}
\label{sigma r}
\phi_{r_1,r_2} : \mathbb{C}^* \hookrightarrow T_r, \qquad t \mapsto (1,1,\underbrace{1,\dots,1}_{r_1 \ 1 \text{'s}},\underbrace{t,\dots,t}_{r_2 \ t\text{'s}})
\end{equation}
The fixed point set with respect to this torus is simply given by those rank $r$ framed sheaves which split as a direct sum of two framed sheaves, of respective ranks $r_1$ and $r_2$:
$$
\MM(r_1) \times \MM(r_2) \cong \MM(r)^{\phi_{r_1,r_2}}, \qquad (\MF_1, \MF_2) \mapsto \MF_1 \oplus \MF_2
$$
In terms of connected components indexed by the second Chern class $n$, we have:
$$
\bigsqcup_{n_1+n_2=n} \MM(r_1)_{n_1} \times \MM(r_2)_{n_2} \cong \MM(r)_n^{\phi_{r_1,r_2}}
$$ 

\end{ex}

\subsection{Symplectic forms and tangent spaces} 
\label{sub:symplectic}

The $K$-theoretic stable basis construction (\cite{AO,O1,O2,OS}, inspired by \cite{MO}), which we will now recall, exists for a wide class of so-called conical symplectic resolutions $X$ \footnote{This class contains all Nakajima quiver varieties, of which the moduli spaces $\MM(r)_n$ are specific examples.}. Such algebraic varieties $X$ are smooth and equipped with a non-degenerate symplectic form:
\begin{equation}
\label{iso tan 1}
\omega : \text{Tan } X \otimes \text{Tan } X \rightarrow \MO_X
\end{equation}
Non-degeneracy implies that \eqref{iso tan 1} yields an isomorphism between the tangent and cotangent bundles. Moreover, we assume that there exists a torus $T$ acting on $X$, which naturally splits into two parts:
\begin{equation}
\label{decomp}
T = \mathbb{C}^*_\omega \times A
\end{equation}
where the subtorus $\mathbb{C}^*_\omega$ scales the symplectic form with weight 1, while the subtorus $A$ preserves the symplectic form. This means that the aforementioned isomorphism between the tangent and cotangent bundles (induced by \eqref{iso tan 1}) is not $T$-equivariant, but becomes $T$-equivariant upon twisting with the character $q : T \rightarrow \mathbb{C}^*$ that sends $(t,a) \mapsto t$ with respect to the decomposition \eqref{decomp}:
\begin{equation}
\label{iso tan 2}
\text{Tan}^\vee X \cong q \otimes \text{Tan }X
\end{equation}
The non-degeneracy of the form \eqref{iso tan 1} implies that the dimension of $X$ = the rank of $\text{Tan } X$ must be even. Moreover, the $T$-equivariant nature of the symplectic form of \eqref{iso tan 2} implies that we have isomorphisms:
$$
\text{Tan}^\vee_p X \cong q \otimes \text{Tan}_p X
$$
between the tangent and cotangent spaces at any $p \in X^T$. Therefore, the weights that appear in the $T$-representation $\text{Tan}_p X$ can be paired up with respect to the symplectic form, and we conclude that:
$$
\chi_T(\text{Tan}_p X) = \sum_{i=1}^{\frac {\dim X}2} \left( \chi_i + \frac 1{q\chi_i} \right)
$$
for various $T$-characters $\chi_i : T \rightarrow \mathbb{C}^*$. Here and below, if $V$ is any representation of the torus $T$, we will write $\chi_T(V) \in \text{Rep}_T$ for the character of $T$ in $V$. Explicity, $\chi_T(V)$ will be a Laurent polynomial in the elementary characters of $T$ (these are denoted by $q_1,q_2,u_1,\dots,u_r$ for the torus $T_r$ of Subsection \ref{sub:torus acts on M}).

\begin{ex} 
\label{fixed point moduli}

The discussion above applies to the moduli spaces of framed sheaves $X = \MM(r)_n$, which are $2nr$ dimensional as algebraic varieties over $\mathbb{C}$ (this generalizes the well-known fact that the Hilbert scheme of $n$ points has dimension $2n$). We consider the action of the torus:
$$
T = T_r = \mathbb{C}^* \times \mathbb{C}^* \times (\mathbb{C}^*)^r
$$
that was considered in Subsection \ref{sub:torus acts on M}, and the decomposition \eqref{decomp} has:
$$\mathbb{C}^*_\omega = \{(t,t,1,\dots,1)\} \hookrightarrow T_r \hookleftarrow A = \{(t,t^{-1},\xi_1,\dots,\xi_r)\}
$$
In particular, the weight of the symplectic form is $q_1q_2$. As for the $T_r$-character in the tangent space to a fixed point of the form \eqref{monomial sheaf}, if we let $q_1,q_2,u_1,\dots,u_r$ denote the usual elementary characters of the factors of $T_r$, we have:
\begin{equation}
\label{explicit tangent space}
\chi_{T_r} \left(\emph{Tan}_{\MF_{\lambda^{1},\dots,\lambda^r}} \MM(r) \right) = \sum_{i,j=1}^r \sum_{\square  \in \lambda^i} \left( \frac {u_j}{u_i} q_1^{-a_{\square}(\lambda^i)-1}q_2^{l_{\square}(\lambda^j)} + \frac {u_i}{u_j} q_1^{a_{\square}(\lambda^i)}q_2^{-l_{\square}(\lambda^j)-1} \right)
\end{equation}
where $\square \in \lambda$ means that $square$ goes over all the boxes in the Young diagram of the partition $\lambda$, and the arm (respectively leg) length $a_{\square}(\mu)$ (respectively $l_{\square}(\mu)$) refers to the number of unit steps one must take to the right (respectively above) of the box $\square$ in order to reach the vertical (respectively horizontal) boundary of the Young diagram $\lambda$. We remark that the arm (respectively leg) length can be negative if the aforementioned boundary is left (respectively below) of the box $\square$. 

\end{ex}

In the case $r=1$, namely the Hilbert scheme of points, formula \eqref{explicit tangent space} reads:
$$
\chi_{\mathbb{C}^* \times \mathbb{C}^*} \left(\text{Tan}_{\lambda} \Hilb \right) = \sum_{\square  \in \lambda} \left( q_1^{-a_{\square}(\lambda)-1}q_2^{l_{\square}(\lambda)} +  q_1^{a_{\square}(\lambda)}q_2^{-l_{\square}(\lambda)-1} \right)
$$
Then formula \eqref{eqn:skyscraper restriction} implies the following restriction formula in $K$-theory:
\begin{equation}
    \label{eqn:i lambda pair}
\mathbb{C}_{I_\lambda}|^L_{I_\lambda} = \prod_{\square  \in \lambda} \left(1 - q_1^{a_{\square}(\lambda)+1}q_2^{-l_{\square}(\lambda)} \right) \left(1 - q_1^{-a_{\square}(\lambda)}q_2^{l_{\square}(\lambda)+1} \right)
\end{equation}
where the left-hand side is naturally interpreted as the torus character in a chain complex of $\mathbb{C}^* \times \mathbb{C}^*$ representations. Note that the right-hand side of the expression above matches \eqref{eqn:b lambda}, which is supported by the prediction that the isomorphism \eqref{identify hilb} sends $[\lambda] = [\mathbb{C}_{I_\lambda}]$ to the modified Macdonald polynomial $\widetilde{H}_\lambda$ (the minor discrepancies between \eqref{eqn:b lambda} and \eqref{eqn:i lambda pair}, namely the substitution $q_1,q_2 \mapsto q,t^{-1}$ and the fact that half of the linear factors need to be moved from the numerator to the denominator, are accounted for in the difference between Macdonald polynomials $P_\lambda$ and modified Macdonald polynomials $\widetilde{H}_\lambda$). 

\subsection{Attracting subvarieties} 
\label{sub:attracting}

For a smooth algebraic variety $X$ endowed with the action of an algebraic torus $T$, we can perform the following construction. For any one-parameter torus $\sigma : \mathbb{C}^* \rightarrow T$ and any connected component of the fixed locus:
\begin{equation}
\label{connected}
Z \subset X^\sigma
\end{equation}
we may define its attracting set:
\begin{equation}
\label{attracting}
\text{Attr}_\sigma(Z) = \{x \in X \text{ s.t. } \lim_{t \rightarrow 0} \sigma(t)\cdot x \text{ exists and lies in }Z\}
\end{equation}
The repelling set is defined analogously, but replacing $\sigma$ with $\sigma^{-1}$. One may define a partial order on the set of connected components of $X^\sigma$ generated by:
$$
Z \succeq Z' \quad \text{if} \quad \overline{\text{Attr}_\sigma(Z)} \cap Z' \neq \emptyset
$$
and transitivity. This allows us to define the full attracting set of a connected component \eqref{connected} as:
\begin{equation}
\label{full attracting}
\text{Attr}^f_\sigma(Z) = \bigsqcup_{Z' \preceq Z} \text{Attr}_\sigma(Z')
\end{equation}
While attracting sets \eqref{attracting} need not be subvarieties of $X$ in general (the reason is that they are only locally closed, instead of closed, subsets of $X$ in the Zariski topology), the full attracting set \eqref{full attracting} will be. This is the very point of the definition of full attracting subvarieties: if one defines the notion of a fixed point ``flowing'' into another fixed point when a small perturbation of the former literally flows into the latter under the action of the one-parameter subtorus $\sigma$, then the full attracting subvariety of $Z$ consists of all points of $X$ which flow into points which flow into points which ... flow into points of $Z$.

\subsection{Normal bundles} 
\label{sub:normal}

Let us now apply the situation in the previous Subsection when $X$ is a symplectic variety, and $\sigma : \mathbb{C}^* \rightarrow A$ is a one-parameter subtorus that preserves the symplectic form. In this case, consider any connected component of the fixed point locus:
$$
Z \subset X^A
$$
The normal bundle to $Z$ will split as a direct sum:
\begin{equation}
\label{normal}
N_Z X = N_Z^+ X \oplus N_Z^-X
\end{equation}
where $N_Z^+X$ (respectively $N_Z^-X$) is the normal bundle of the repelling (respectively attracting) subvariety of $Z$. In other words, the weights of the torus $A$ acting in the fibers of $N_Z^+X$ (respectively $N_Z^-X$) are precisely those weights which are positive (respectively negative) with respect to the cocharacter $\sigma$. \footnote{Given a cocharacter $\sigma : \mathbb{C}^* \rightarrow A$, a character $\chi : A \rightarrow \mathbb{C}^*$ is called positive (respectively negative) if $\chi \circ \sigma(t) = t^n$ for $n$ a positive (respectively negative) integer. The notion applies similarly if $\chi : T \rightarrow \mathbb{C}^*$ for a bigger torus $T \supset A$.} The symplectic form pairs $N_Z^+X$ and $N_Z^-X$ non-trivially, and so we have the following analogue of \eqref{iso tan 2}:
$$
(N_Z^+X)^\vee \cong q \otimes N_Z^- X
$$
as $T$-equivariant vector bundles on $Z$. In particular, the dimensions of the attracting and repelling normal bundles in \eqref{normal} are equal to each other, which can be thought of as saying that the attracting/repelling subvarieties are half-codimensional with respect to the subvariety $Z \subset X$ (in particular, this implies that the dimension of $Z$ must be even).

\begin{ex} Let us consider the situation of Example \ref{fixed point moduli}, when $X = \MM(r)_n$, and the fixed point set $X^A$ consists of finitely many points:
\begin{equation}
\label{fixed point p}
p = \MF_{\lambda^1,\dots,\lambda^r}
\end{equation}
In this case, the normal bundle to such a fixed point is simply the tangent space at $p$, whose $T_r$-character is given by formula \eqref{explicit tangent space}. For any cocharacter $\sigma : \mathbb{C}^* \rightarrow A \subset T_r$, one of the two terms:
\begin{equation}
\label{two terms}
\frac {u_j}{u_i} q_1^{-a_{\square}(\lambda^i)-1}q_2^{l_{\square}(\lambda^j)} \qquad \text{and} \qquad \frac {u_i}{u_j} q_1^{a_{\square}(\lambda^i)}q_2^{-l_{\square}(\lambda^j)-1}
\end{equation}
will be positive and the other will be negative with respect to $\sigma$. Therefore, we have:
$$
\emph{Tan}_p X = \emph{Tan}^+_p X \oplus \emph{Tan}^-_p X
$$
where the $T_r$-character of $\emph{Tan}^+_p X$ (respectively $\emph{Tan}^-_p X$) is simply the sum of those terms among \eqref{two terms} which are positive (respectively negative) with respect to $\sigma$. Special choices of $\sigma$ will make it very obvious which is the positive term and which is the negative term. For example, if:
\begin{equation}
\label{one-parameter}
\sigma_{d_1,\dots,d_r}(t) = (t,t^{-1}, t^{d_1}, \dots ,t^{d_r})
\end{equation}
for integers $d_1 \ll \dots \ll d_r$ (the inequalities should be interpreted as saying that $d_i$ is much smaller than $d_{i+1}$ compared to the weight $n$ of the partitions that make up the fixed point \eqref{fixed point p}), then the first term in \eqref{two terms} will be the positive one and the second term will be the negative one whenever $i < j$.

\end{ex}

\subsection{The stable basis} 
\label{sub:stable}

We are ready to give the definition of the $K$-theoretic stable basis (\cite{AO,O1,O2,OS}, inspired by \cite{MO}), at least modulo certain details. For simplicity, we will only deal with symplectic algebraic varieties $X$ endowed with $T$-actions for which all the preceding discussion applies, and whose fixed point set $X^T$ is finite. Our main example is $X = \MM(r)_n$, which satisfies all of these properties. 

\begin{defn} 
\label{def stable basis}

For any generic cocharacter $\sigma : \mathbb{C}^* \rightarrow A$, there is a collection of elements:
\begin{equation}
\label{stable basis}
\{s_p\}_{p \in X^T} \in K^T(X)
\end{equation}
which are uniquely determined by the following properties for all $p \in X^T$:

\begin{enumerate}

\item $s_p$ is supported on the full attracting subvariety $\emph{Attr}^f_\sigma(p)$

\item $s_p|^L_p = \wedge^\bullet(\emph{Tan}^{-,\vee}_p X) \otimes \emph{monomial}_p$

\item for any fixed point $p' \prec p$, the $A$-weights of $s_p|^L_{p'}$ are contained (as elements of $\emph{Lie}_{\mathbb{R}}(A^\vee)$) in:
\begin{equation}
\label{window}
\Big(\text{convex hull}^*\text{ of weights appearing in } \wedge^\bullet(\emph{Tan}^-_{p'}X) \Big) + \emph{shift}_{p,p'} \subset \emph{Lie}_{\mathbb{R}}(A^\vee)
\end{equation}
where $\text{convex hull}^*$ refers to the convex hull minus one of its vertices.  

\end{enumerate}

\end{defn}

\begin{rmk}
\label{rem:choices}

The (so-far undefined) quantities $\emph{monomial}_p$ and $\emph{shift}_{p,p'}$ are certain characters of $T$ and $A$, respectively. The definition of these quantities takes as input two more pieces of data that one must specify in order to give the full definition of the stable basis, namely: 

\begin{itemize}

\item a polarization, i.e. the choice of a decomposition of the tangent spaces to $X$ into halves

\item a $T$-equivariant line bundle $\ML$ on $X$

\end{itemize}

The interested reader may find a complete discussion of these issues in \cite[Section 9]{O lect}.  We do not dwell on them, because these choices are quite natural in the case at hand: because the moduli space of sheaves $X = \MM(r)$ is the cotangent bundle of a certain stack $S$, $\emph{Tan }\MM(r)$ is locally isomorphic to $\emph{Tan } S \oplus \emph{Tan}^\vee S$, and this decomposition into ``halves of the tangent bundle'' yields the requisite polarization. As for the line bundle $\ML$, it will always be taken to be the structure sheaf $\MO_X$ in the constructions that follow.

\end{rmk}

Properties (1) and (2) imply that the collection \eqref{stable basis} is upper triangular in the fixed point basis $[p]$:
\begin{equation}
\label{stable to fixed}
s_p = [p] \cdot \frac {\text{monomial}_p}{\wedge^\bullet(\text{Tan}^{+,\vee}_p X)} + \sum_{p' \prec p} [p'] \cdot \text{coefficient}
\end{equation}
where the numerator of the fraction is a character of $T$, i.e. a monomial in the representation ring $\text{Rep}_T$, and the quantities marked ``coefficient'' are elements of $\text{Frac}(\text{Rep}_T)$. Therefore, the elements \eqref{stable basis} form a basis of $K^T(X)_{\text{loc}}$, which is called the stable basis. Property (3) implies the uniqueness of the stable basis, and also provides an algorithm for computing it via long division of Laurent polynomials in several variables: fix a total ordering of the fixed points $p_1,\dots,p_N$ which refines the partial ordering $\prec$. Then we start from the collection $\{t_p = [\MO_{\text{Attr}^f_p(X)}] \otimes \text{monomial}_p\}_{p \in X^T} \in K^T(X)$ and at the $i$-th step, we will modify: 
$$
t_{p_i} \leadsto s_{p_i} = t_{p_i} + \sum_{j=1}^{i-1} \text{coefficient} \cdot s_{p_j}
$$
in such a way that property (3) holds for $p = p_i$. The fact that the already constructed $s_{p_1},\dots,s_{p_{i-1}}$ already satisfy properties (2) and (3) implies that the coefficients in the above formula are completely determined. The existence of the stable basis satisfying properties (1)-(3) was established in \cite{AO,O1,O2}.

\subsection{The geometric $R$-matrix}
\label{sub:r-matrix}

Let us now set $X = \MM(r)_n$, and consider the one-parameter subgroup \eqref{one-parameter} for $d_1 \ll \dots \ll d_r$. The construction of the previous Subsection gives rise to a stable basis:
\begin{equation}
\label{attracting basis}
\{s^+_{\lambda^1,\dots,\lambda^r}\}_{\text{partitions } \lambda^1, \dots,\lambda^r} \in K^{T_r}(\MM(r))
\end{equation}
Similarly, the analogous construction for the inverse one-parameter subgroup $\sigma^{-1}_{d_1,\dots,d_r}$ yields a basis:
\begin{equation}
\label{repelling basis}
\{s^-_{\lambda^1,\dots,\lambda^r}\}_{\text{partitions } \lambda^1, \dots,\lambda^r} \in K^{T_r}(\MM(r))
\end{equation}
In other words, the basis \eqref{repelling basis} is defined just like the basis \eqref{attracting basis}, but switching the roles of attracting and repelling subvarieties. This implies that the bases $s^+_{\lambda^1,\dots,\lambda^r}$ and $s^-_{\lambda^1,\dots,\lambda^r}$ are upper and lower triangular, respectively, in the basis of fixed points $[\lambda^1,\dots,\lambda^r]$ with respect to the ordering $\prec$. 

\begin{defn} 

For any decomposition $r = r_1+r_2$, we have stable basis maps:
\begin{equation}
\label{stable basis map}
K^{T_{r_1}}(\MM(r_1)) \otimes K^{T_{r_2}}(\MM(r_2)) \stackrel{\emph{Stab}^\pm}\longrightarrow K^{T_r}(\MM(r))
\end{equation}
given by:
$$
\emph{Stab}^\pm \left( s^\pm_{\lambda^1,\dots,\lambda^{r_1}} \otimes s^\pm_{\mu^1,\dots,\mu^{r_2}} \right) = s^\pm_{\lambda^1,\dots,\lambda^{r_1},\mu^{1},\dots,\mu^{r_2}}
$$
for all partitions $\lambda^1,\dots,\lambda^{r_1},\mu^1,\dots,\mu^{r_2}$. The maps \eqref{stable basis map} are isomorphisms after tensoring with $\mathbb{F}_r$.

\end{defn}

In the usual stable basis framework of \cite{MO}, it is natural to think of the maps \eqref{stable basis map} as the fundamental construction, instead of the elements \eqref{attracting basis}--\eqref{repelling basis}. Indeed, Definition \ref{def stable basis} can be upgraded to the language of Lagrangian correspondences, in which the subvariety $\MM(r_1) \times \MM(r_2)$ plays the role of individual torus fixed points of $\MM(r)$ (and in fact, the aforementioned subvariety is the torus fixed locus with respect to the one-parameter subgroup \eqref{sigma r}). An overview of the theory in this more general language can be found in \cite[Section 9]{O lect}.

\begin{defn} 

For any decomposition $r = r_1+r_2$, the composition:
\begin{equation}
\label{r-matrix}
R_{r_1,r_2} : K^{T_{r_1}}(\MM(r_1)) \otimes K^{T_{r_2}}(\MM(r_2)) \xrightarrow{\emph{Stab}^+} K^{T_r}(\MM(r)) \xrightarrow{(\emph{Stab}^-)^{-1}} K^{T_{r_1}}(\MM(r_1)) \otimes K^{T_{r_2}}(\MM(r_2))
\end{equation}
is called a geometric $R$-matrix. It satisfies the quantum Yang-Baxter equation:
\begin{equation}
\label{qybe}
R_{r_1,r_2} R_{r_1,r_3} R_{r_2,r_3} = R_{r_2,r_3} R_{r_1,r_3} R_{r_1,r_2}
\end{equation}
as endomorphisms of $K^{T_{r_1}}(\MM(r_1)) \otimes K^{T_{r_2}}(\MM(r_2)) \otimes K^{T_{r_3}}(\MM(r_3))$, for all natural numbers $r_1,r_2,r_3$. \footnote{Implicit in the notation \eqref{qybe} is that $R_{r_1,r_3}$ acts as \eqref{r-matrix} as an endomorphism of $K(\MM(r_1)) \otimes K(\MM(r_3))$ tensor the identity on $K(\MM(r_2))$, etc.}

\end{defn}

\subsection{Poles of the $R$-matrix}
\label{sub:poles}

As we have seen in \eqref{identify hilb} and \eqref{identify m}, matching modified Macdonald polynomials with the skyscraper sheaves at the torus fixed points allows us to identify $K^{T_r}(\MM(r))$ with a tensor product of Fock spaces. Under this isomorphism, the $R$-matrix \eqref{r-matrix} gives rise to an endomorphism:
\begin{equation}
\label{r-matrix in fock}
R_{r_1,r_2} \in \text{End}_{\mathbb{F}_r} \left( \Lambda_{\mathbb{F}}(u_1) \otimes \dots \otimes \Lambda_{\mathbb{F}}(u_{r}) \right)
\end{equation}
where $r = r_1+r_2$. Since the endomorphism above is ``geometric'' in nature (the rigorous term here is that it is given by a correspondence), all its matrix coefficients are rational functions in $q_1,q_2,u_1,\dots,u_r$. However, the very nature of the stable basis allows us to say more.

\begin{prop} The poles of the endomorphism $R_{r_1,r_2}$ are all of the form:
$$
u_j = u_i q_1^x q_2^y
$$
as $1 \leq i \leq r_1 < j \leq r$ and $x,y \in \mathbb{Z}$.

\end{prop}

\begin{proof} As explained in Subsection 9.3.5 of \cite{O lect}, the endomorphism $R_{r_1,r_2}$ is a product of endomorphisms $R_{1,1}$ acting in various tensor products of the form $\Lambda_{\mathbb{F}}(u_i) \otimes \Lambda_{\mathbb{F}}(u_j)$ for $1 \leq i \leq r_1 < j \leq r$ (and the identity on the other tensor factors) and so it is enough to prove the Proposition in the case $r_1 = r_2 = 1$. As is clear from \eqref{stable basis map} and \eqref{qybe}, when expressed in the basis $s^+_\lambda \otimes s^+_\mu$, the endomorphism $R_{1,1}$ is the product of the following matrices:

\begin{itemize}

\item the change of basis from $s^+_{\lambda,\mu}$ to $s^-_{\lambda,\mu}$ as $\lambda$ and $\mu$ go over all partitions

\item the change of basis from $s^-_\lambda \otimes s^-_\mu$ to $s^+_\lambda \otimes s^+_\mu$ as $\lambda$ and $\mu$ go over all partitions

\end{itemize}

The second bullet is simply the tensor product of a matrix acting in $\Lambda_{\mathbb{F}}(u_1)$ and a matrix acting in $\Lambda_{\mathbb{F}}(u_2)$, so it does not produce any poles involving $u_1$ and $u_2$. As for the first bullet, we can further sub-divide it into:

\begin{itemize}

\item the change of basis from $s^+_{\lambda,\mu}$ to $[\lambda,\mu]$ as $\lambda$ and $\mu$ go over all partitions

\item the change of basis from $[\lambda,\mu]$ to $s^-_{\lambda,\mu}$ as $\lambda$ and $\mu$ go over all partitions

\end{itemize} 

As shown by \eqref{stable to fixed}, the poles of the aforementioned changes of basis are among the linear factors in:
$$
\wedge^\bullet(\text{Tan}_{\MF_{\lambda,\mu}}^\vee \MM(2))
$$
From formula \eqref{explicit tangent space}, those linear factors which involve both $u_1$ and $u_2$ in the expression above are:
$$
1 - \frac {u_1}{u_2} q_1^x q_2^y \quad \text{and} \quad 1 - \frac {u_2}{u_1} q_1^x q_2^y
$$
for various $(x,y) \in \mathbb{Z}^2$ which arise as the weights of boxes in various partitions. 

\end{proof}

A finer analysis of the poles of the $R$-matrix via shift operators (following the lines of \cite[Subsections 3.4 and 3.5]{K}) allows one to prove that the poles of $R_{r_1,r_2}$ are simple and all of the form $u_j = u_i q_1^x q_2^y$ as $1 \leq i \leq r_1 < j \leq r$ and $x,y \in \mathbb{N}$. We thank Yakov Kononov and Andrey Smirnov for this remark. 

\subsection{Connection between the $R$-matrices}

As the reader probably suspects at this stage, the endomorphisms \eqref{r-matrix in fock} are expected to match the ones produced by the $R$-matrix of the quantum toroidal algebra. 

\begin{conj}
\label{conj:two r-matrices}

The endomorphism $R_{1,1} \in \emph{End}_{\mathbb{F}_2} \left( \Lambda_{\mathbb{F}}(u_1) \otimes\Lambda_{\mathbb{F}}(u_{2}) \right)$ matches $R(u_2/u_1)$ of \eqref{eqn:r-matrix in fock}. 

\end{conj}

As we already mentioned in Remark \ref{rem:choices}, part of the Conjecture above involves fixing the various choices that go into the Definition of the stable basis, with the goal of having it match $R(u_2/u_1)$ on the nose. In fact, these choices are quite strongly determined by the existence of an action:
\begin{equation}
\label{eqn:geometric action}
U_{q_1,\frac 1{q_2}}(\ddot{\mathfrak{gl}}_1) \Big |_{c = 1, \psi_0^+ = 1, \psi_0^- = (q_1q_2)^{-r}} \curvearrowright K^{T_r}(\MM(r))_{\text{loc}}
\end{equation}
which generalizes \eqref{eqn:action} under the identification \eqref{identify hilb}. A way to reformulate the actions above (as $r$ runs over $\mathbb{N}$) is to frame them as an algebra homomorphism:
\begin{equation}
\label{eqn:algebra homomorphism}
\Upsilon : U_{q_1,\frac 1{q_2}}(\ddot{\mathfrak{gl}}_1) \Big |_{c = 1, \psi_0^+ = 1, \psi_0^- = (q_1q_2)^{-r}} \longrightarrow \prod_{r=1}^{\infty} \text{End}_{\mathbb{F}_r} \left( K^{T_r}(\MM(r))_{\text{loc}} \right)
\end{equation}
(implicit in the formula above is that ``$r$'' denotes the grading element with respect to the product in the right-hand side). On the other hand, the usual FRT formalism (explained in \cite{MO} as pertains to our context) says that once one has $R$-matrices \eqref{r-matrix} for all $r_1$ and $r_2$ that satisfy the quantum Yang-Baxter equation, taking arbitrary matrix coefficients of these $R$-matrices gives rise to a Hopf subalgebra:
\begin{equation}
    \label{eqn:frt}
\mathcal{U} \subset \prod_{r=1}^{\infty} \text{End}_{\mathbb{F}_r} \left( K^{T_r}(\MM(r))_{\text{loc}} \right)
\end{equation}
Tautologically, $R_{r_1,r_2}$ are the images of the universal $R$-matrix of $\mathcal{U}$ in the representations $K^{T_r}(\MM(r))_{\text{loc}}$.

\begin{conj} 
\label{conj:big}

The map $\Upsilon$ of \eqref{eqn:algebra homomorphism} induces a Hopf algebra isomorphism:
\begin{equation}
\label{eqn:big}
U_{q_1,\frac 1{q_2}}(\ddot{\mathfrak{gl}}_1) \Big |_{c = 1, \psi_0^+ = 1, \psi_0^- = (q_1q_2)^{-r}}  \xrightarrow{\sim} \mathcal{U}
\end{equation}

\end{conj} 

There are several parts to proving Conjecture \ref{conj:big}, and we will list them in increasing order of difficulty.

\begin{enumerate}[leftmargin=*]

\item Show that the image of the map $\Upsilon$ lands in the subalgebra $\mathcal{U}$ of \eqref{eqn:frt}. This is actually fairly easy, since the action \eqref{eqn:geometric action} is generated by tautological line bundles on the so-called simple Nakajima correspondences (see \cite{FT,SV}, also the proof of Proposition \ref{prop:inj}); these correspondences are Lagrangian, and they fit in quite well within the framework of stable bases.

\item Show that \eqref{eqn:big} is not just a map of algebras, but a map of bialgebras. To this end, one needs to show that the map \eqref{eqn:algebra homomorphism} intertwines the coproduct \eqref{eqn:cop quantum 0}--\eqref{eqn:cop quantum 4} on the left-hand side with the coproduct induced by stable basis maps (see \cite[Section 3.3]{OS}) on the right-hand side.

\item Show that the map \eqref{eqn:big} is a bijection: we will deal with injectivity in the following Proposition, but note that surjectivity is the more significant task.

\end{enumerate}

\begin{prop}
\label{prop:inj}

The map \eqref{eqn:algebra homomorphism} is injective.

\end{prop}

\begin{proof}

There is a triangular decomposition which is orthogonal to \eqref{eqn:triangular} (see \cite{R-matrix} for details):
\begin{equation}
\label{eqn:orthogonal}
U_{q_1,\frac 1{q_2}}(\ddot{\mathfrak{gl}}_1) \Big |_{c = 1, \psi_0^+ = 1, \psi_0^- = (q_1q_2)^{-r}} = \MS \otimes \mathbb{F}[\psi_k^{\pm}]_{k \in \mathbb{N}} \otimes \MS^{\text{op}}
\end{equation}
such that the three factors in the decomposition above act on:
$$
K^{T_r}(\MM(r))_{\text{loc}} = \bigoplus_{n = 0}^{\infty} K^{T_r}(\MM(r)_n)_{\text{loc}}
$$
by decreasing $n$, preserving $n$ and increasing $n$, respectively. More specifically:
\begin{equation}
\label{eqn:psi acts}
\psi^\pm(z) \cdot [\lambda^1,\dots,\lambda^r] = [\lambda^1,\dots,\lambda^r] \cdot \prod_{a=1}^r \left[ \frac {z - u_a q_1^{-1}q_2^{-1}}{z-u_a} \prod_{\square = (i,j) \in \lambda^a} \frac {\zeta \left(\frac {u_a q_1^{i}q_2^{j}}z \right)}{\zeta \left(\frac {u_a q_1^{i-1}q_2^{j-1}}z \right)} \right]
\end{equation}
where $\psi^\pm (z) = \sum_{k=0}^{\infty} \psi^\pm_k z^{\mp k}$. In the formula above, $\zeta$ denotes the rational function \eqref{eqn:def zeta} with $q \mapsto q_1$, $t \mapsto q_2^{-1}$. Since $[\lambda^1,\dots,\lambda^r]$ correspond to tensor products of modified Macdonald polynomials according to \eqref{identify m}, then the polynomial ring generated by the operators $\psi_k^\pm$ coincides with the polynomial ring generated by the (plethystically modified versions of the) Macdonald operators \eqref{En} and \eqref{En_dual}. In particular, it is easy to deduce from \eqref{eqn:psi acts} that the elements $\{h_{\pm k}\}_{k \in \mathbb{N}} \in \qg$ of \eqref{eqn:psi to h} act by:
\begin{equation}
\label{eqn:h acts}
h_{\pm k} \cdot [\lambda^1,\dots,\lambda^r] = [\lambda^1,\dots,\lambda^r] \cdot \pm \sum_{a=1}^r \left(\frac {u_a^{\pm k}}{(1-q_1^{\pm k})(1-q_2^{\pm k})} - \sum_{\square = (i,j) \in \lambda^a} u_a^{\pm k}q_1^{\pm (i-1)k} q_2^{\pm (j-1)k} \right)
\end{equation}
It is sometimes convenient to replace the operators above by:
\begin{multline*}
h_{\pm k}' = -h_{\pm k} \pm \text{Id}_{K^{T_r}(\MM(r))_{\text{loc}}} \cdot \sum_{a=1}^r \frac {u_a^{\pm k}}{(1-q_1^{\pm k})(1-q_2^{\pm k})} \qquad \Rightarrow \\ \Rightarrow \qquad h_{\pm k}' \cdot [\lambda^1,\dots,\lambda^r] = [\lambda^1,\dots,\lambda^r] \cdot \sum^{1\leq a \leq r}_{\square = (i,j) \in \lambda^a} u_a^{\pm k}q_1^{\pm (i-1)k} q_2^{\pm (j-1)k} 
\end{multline*}
Similarly, the tensor factors $\MS$ and $\MS^{\text{op}}$ act by the following formulas:
\begin{align}
&F \cdot [\mu^1,\dots,\mu^r] = \mathop{\sum_{\{\mu^a \unlhd \lambda^a\}_{a = 1, \dots, r}}}_{\sum_{a=1}^r |\lambda^a \backslash \mu^a| = n} [\lambda^1,\dots,\lambda^r] \cdot F( \dots , u_a q_1^{i-1} q_2^{j-1}, \dots)_{\square = (i,j) \in \lambda^a \backslash \mu^a}^{a \in 1,\dots,r} \cdot \text{factor}_{\mu^1,\dots,\mu^r} \label{eqn:action 1} \\
&G \cdot [\lambda^1,\dots,\lambda^r] = \mathop{\sum_{\{\mu^a \unlhd \lambda^a\}_{a = 1, \dots, r}}}_{\sum_{a=1}^r |\lambda^a \backslash \mu^a| = n}  [\mu^1,\dots,\mu^r] \cdot G( \dots , u_a q_1^{i-1} q_2^{j-1}, \dots)_{\square = (i,j) \in \lambda^a \backslash \mu^a}^{a \in 1,\dots,r} \cdot \text{factor}'_{\lambda^1,\dots,\lambda^r} \label{eqn:action 2}
\end{align}
for any $F(z_1,\dots,z_n)\in \MS$ and $G(z_1,\dots,z_n) \in \MS^{\text{op}}$, where $\mu \unlhd \lambda$ means that the Young diagram of the partition $\mu$ is contained in the Young diagram of the partition $\lambda$. In formulas \eqref{eqn:action 1}--\eqref{eqn:action 2}, the terms denoted $\text{factor}_{\mu^1,\dots,\mu^r}$ and $\text{factor}'_{\lambda^1,\dots,\lambda^r}$ are certain non-zero elements of $\mathbb{F}$ that the interested reader may find in \cite[Subsection 3.10]{AGT}; the important thing for now is that they do not depend on $F$ and $G$, respectively. We conclude from the formulas above that $F(z_1,\dots,z_n) \in \MS$ (respectively $G(z_1,\dots,z_n) \in \MS^{\text{op}}$) acts on $r$-tuples of partitions by adding (respectively removing) $n$ boxes; the respective matrix coefficient is proportional to $F$ (respectively $G$) applied to the collection of weights $\{u_aq_1^{i-1}q_2^{j-1}\}_{\square = (i,j)}$ of the added (respectively removed) boxes. 

With this description of the action \eqref{eqn:geometric action} in mind, let us prove the injectivity of the map \eqref{eqn:algebra homomorphism}. By \eqref{eqn:orthogonal}, we need to show that for an arbitrary non-zero tensor:
\begin{equation}
\label{eqn:tensor}
T = \sum_i \underbrace{F_i}_{\in \MS} \otimes \underbrace{\chi_i}_{\in \mathbb{F}[h_{\pm k}']} \otimes \underbrace{G_i}_{\in \MS^{\text{op}}} 
\end{equation}
there exists a large enough $r$ and collections of partitions $\lambda^1,\dots,\lambda^r$, $\mu^1,\dots,\mu^r$ such that the coefficient: 
\begin{equation}
\label{eqn:the coefficient}
\left \langle [\mu^1, \dots, \mu^r] |  T  | [\lambda^1, \dots, \lambda^r] \right \rangle
\end{equation}
is non-zero. To this end, choose any $n_1,n_2 \in \mathbb{N}$ and $r \gg n_1 + n_2$. Then we consider the following $r$-tuples of partitions, all of which either consist of a single box or are empty:
\begin{align*} 
&(\lambda^1,\dots,\lambda^r) = (\underbrace{\emptyset, \dots, \emptyset}_{n_1 \text{ terms}}, \underbrace{\square, \dots, \square}_{r-n_1-n_2 \text{ terms}}, \underbrace{\square, \dots, \square}_{n_2 \text{ terms}}) \\
&(\mu^1,\dots,\mu^r) = (\underbrace{\square, \dots, \square}_{n_1 \text{ terms}}, \underbrace{\square, \dots, \square}_{r-n_1-n_2 \text{ terms}}, \underbrace{\emptyset, \dots, \emptyset}_{n_2 \text{ terms}})
\end{align*}
Then let us choose $n_1$ and $n_2$ to be the maximal number of variables among the shuffle elements $F_i$ and $G_i$ (respectively) which appear in formula \eqref{eqn:tensor}. The non-zero contributions to the matrix coefficient \eqref{eqn:the coefficient} all arise from having the $G_i$'s remove the last $n_2$ boxes from the $r$-tuple of partitions $(\lambda^1,\dots,\lambda^r)$ and have the $F_i$'s add the first $n_1$ boxes, in order to obtain the $r$-tuple of partitions $(\mu^1,\dots,\mu^r)$. With this in mind, the coefficient \eqref{eqn:the coefficient} equals:
$$
\sum_i F_i(u_1,\dots,u_{n_1}) \cdot \chi_i\Big|_{h_{\pm k}' \mapsto u_{n_1+1}^{\pm k} + \dots + u_{r-n_2}^{\pm k}} \cdot G_i(u_{r-n_2+1}, \dots, u_r) \cdot \text{non-zero factor}
$$
where the sum above only goes over those indices $i$ such that $F_i$ has $n_1$ variables and $G_i$ has $n_2$ variables. Since the parameters $u_1,\dots,u_r$ are generic, the non-vanishing of the expression above follows from the non-vanishing of the tensor \eqref{eqn:tensor}, and we are done.

\end{proof}

\begin{rmk}

Besides the algebraic approach outlined above, a more straightforward way to prove Conjecture \ref{conj:two r-matrices} would involve comparing formulas for the two $R$-matrices involved, for example \eqref{Rcc} with the trigonometric version of the product formula of \cite[equation (130)]{S} (we thank Andrey Smirnov for pointing out the latter formula to us).

\end{rmk}



\begin{thebibliography}{10}

\bibitem{AO}
M. Aganagic and A. Okounkov,
\newblock {\em Elliptic stable envelopes}
\newblock J. Amer. Math. Soc. 34 (2021), no. 1, 79--133.

\bibitem{Awata}
H.~Awata, H.~Kanno, A.~Mironov, A.~Morozov, A.~Morozov, Y.~Ohkubo and
  Y.~Zenkevich.
\newblock {\em Toric Calabi-Yau threefolds as quantum integrable systems.
  $\mathrm{\mathcal{R}}$-matrix and
  $\mathrm{\mathcal{R}\mathcal{T}\mathcal{T}}$ relations},
\newblock {J. High Energy Phys.}, 2016(10):47, 2016.


\bibitem{DI} 
J.~Ding, K.~Iohara, 
\newblock {\em Generalization of Drinfeld quantum affine algebras},
\newblock {Lett. Math. Phys.} 41 (1997), no. 2, 181-–193

\bibitem{FHHSY}
B.~Feigin, K.~Hashizume, A.~Hoshino, J.~Shiraishi and S.~Yanagida,
\newblock {\em A commutative algebra on degenerate $\mathbb{C}\mathbb{P}^1$ and
  Macdonald polynomials},
\newblock {J. Math. Phys.}, 50(9):095215, 2009.

\bibitem{FJMM_BA}
B. Feigin, M. Jimbo, T. Miwa and E.  Mukhin,
\newblock {\em Quantum toroidal and Bethe ansatz},
\newblock J. Phys. A: Math. Theor., 48, 24, p 244001, 2015.

\bibitem{FT}
B. Feigin, A. Tsymbaliuk,
\newblock {\em Equivariant $K$-theory of Hilbert schemes via shuffle algebra},
\newblock Kyoto J. Math. 51 (2011), no. 4, 831--854.

\bibitem{Fukud}
M.~Fukuda, K.~Harada, Y.~Matsuo and R.-D. Zhu.
\newblock {{\em The Maulik--Okounkov R-matrix from the Ding--Iohara--Miki algebra}},
\newblock { Prog. Theor. Exp. Phys.}, 2017(9), 2017.


\bibitem{GdG}
A. Garbali, J. de Gier, 
\newblock {\em The $R$-matrix of the quantum toroidal algebra $U_ {q,t}(\overset{..}{gl}_1)$ in the Fock module},
\newblock Commun. Math. Phys., 384.3: 1971--2008.

\bibitem{Ha}
M. Haiman.
\newblock{\em  Macdonald polynomials and geometry},
\newblock, New perspectives in geometric combinatorics (Billera, Bj{\"o}rner, Greene, Simion, and Stanley, eds.), MSRI Publications, vol. 38, Cambridge
University Press, 1999, pp. 207--254.

\bibitem{H}
R. Hartshorne.
\newblock{\em Algebraic geometry},
\newblock Graduate Texts in Mathematics, No. 52. Springer--Verlag, New York---Heidelberg, 1977. xvi+496 pp. ISBN: 0--387--90244--9

\bibitem{K} I. Kononov.
\textit{Elliptic Stable Envelopes and 3D Mirror Symmetry}, Thesis (Ph.D.)–Columbia University. 2021. 91 pp. ISBN: 979-8708-75604-6

\bibitem{Litv20}
A.~Litvinov and I. Vilkoviskiy.
\newblock {\em{Liouville reflection operator, affine Yangian and Bethe ansatz}},
\newblock  J. High Energy Phys., 2020.12: 1-49, 2020.

\bibitem{Macdonald} I. Macdonald.
\textit{Symmetric functions and Hall polynomials}, (2nd ed.), Oxford, Clarendon Press 1995.

\bibitem{MO}
D. Maulik and A. Okounkov,
\newblock {\em Quantum Groups and Quantum Cohomology},
\newblock Asterisque 408, 2019.

\bibitem{Miki}
K. Miki,
\newblock {\em A $(q, \gamma)$ analog of the $W_{1+\infty}$ algebra},
\newblock J. Math. Phys., 48.12, 123520, 2007.

\bibitem{N}
H. Nakajima,
\newblock {\em Lectures on Hilbert schemes of points on surfaces},
\newblock  University Lecture Series, vol. 18, 132 pp, 1999.

\bibitem{Shuf}
A. Negu\cb{t},
\newblock {\em The shuffle algebra revisited},
\newblock Int. Math. Res. Not., Volume 2014, Issue 22, 2014, Pages 6242?6275, https://doi.org/10.1093/imrn/rnt156

\bibitem{Flags}
A. Negu\cb{t},
\newblock {\em Moduli of flags of sheaves and their $K$-theory},
\newblock Algebraic Geometry 2 (1) (2015) 19–43, doi:10.14231/AG-2015-002

\bibitem{AGT}
A. Negu\cb{t},
\newblock {\em  The $q$-AGT-W relations via shuffle algebras},
\newblock Comm. Math. Phys. 358 (2018), no. 1, 101--170

\bibitem{R-matrix}
A. Negu\cb{t},
\newblock {\em The $R$-matrix of the quantum toroidal algebra},
\newblock ar$\chi$iv:2005.14182

\bibitem{O lect}
A. Okounkov,
\newblock {\em Lectures on K-theoretic computations in enumerative geometry},
\newblock Geometry of moduli spaces and representation theory, 251--380, IAS/Park City Math. Ser., 24, Amer. Math. Soc., Providence, RI, 2017

\bibitem{O1}
A. Okounkov,
\newblock {\em Inductive construction of stable envelopes},
\newblock ar$\chi$iv:2007.09094

\bibitem{O2}
A. Okounkov,
\newblock {\em Nonabelian stable envelopes, vertex functions with
descendents, and integral solutions of $q$--difference equations},
\newblock ar$\chi$iv:2010.13217

\bibitem{OS}
A. Okounkov and A. Smirnov,
\newblock {\em Quantum difference equation for Nakajima varieties},
\newblock ar$\chi$iv:1602.09007.


\bibitem{Proch19}
T.~Proch\'azka.
\newblock {\em{Instanton R-matrix and W-symmetry}},
\newblock  J. High Energy Phys., 1903.10372: 1-58, 2019.


\bibitem{SV}
O. Schiffmann and E. Vasserot,
\newblock {\em {The elliptic Hall algebra and the $K$-theory of the Hilbert scheme of $\mathbb{A}^2$}}, 
\newblock Duke Math. J. 162 (2013), no. 2, 279--366.

\bibitem{Sh}
J. Shiraishi,
\newblock {\em {A family of integral transformations and basic hypergeometric series}}, 
\newblock Commun. Math. Phys., 263.2: 439--460, 2006.

\bibitem{S}
A. Smirnov,
\newblock {\em {On the instanton $R$-matrix}}, 
\newblock Commun. Math. Phys.,  345 (2016), no. 3, 703–740.

\bibitem{TT}
R. Thomason,
\newblock {\em Une formule de lefschetz en k-th{\'e}orie {\'e}quivariante alg{\'e}brique},
\newblock Duke Mathematical Journal 68 (1992), no. 3, 447--462.


\end{thebibliography}
\end{document}